\documentclass[10pt,a4paper]{article}
\usepackage[linktocpage]{hyperref} 
\usepackage{amssymb}
\usepackage{amsmath}
\usepackage{amsthm}
\usepackage{latexsym}
\usepackage{mathrsfs}
\usepackage{bm}
\usepackage{authblk}
\usepackage{textalpha}

\theoremstyle{plain}
\newtheorem{theorem}{Theorem}

\newtheorem*{assumption}{Assumption}

\newtheorem{definition}{Definition}
\newtheorem{remark}{Remark}

\def\Qop{\text{\rm\textQoppa}}

\def\bfa{{\bf a}}
\def\bfb{{\bf b}}
\def\bfc{{\bf c}}

\def\bfA{{\bf A}}
\def\bfB{{\bf B}}
\def\bfC{{\bf C}}

\def\usemedskip{}

\begin{document}
\title{\bf A formalism for the calculus of variations with spinors }

\author[,1]{Thomas B\"ackdahl \footnote{E-mail address:{\tt t.backdahl@ed.ac.uk}}}
\author[,2]{Juan A. Valiente Kroon \footnote{E-mail address:{\tt j.a.valiente-kroon@qmul.ac.uk}}}

\affil[1]{The School of Mathematics, University of Edinburgh, JCMB 6228, 
Peter Guthrie Tait Road, Edinburgh EH9 3FD, United Kingdom}

\affil[2]{School of Mathematical Sciences, Queen Mary, University of London,
Mile End Road, London E1 4NS, United Kingdom}

\maketitle

\begin{abstract}
We develop a frame and dyad gauge-independent formalism for the
calculus of variations of functionals involving spinorial objects. As
part of this formalism we define a modified variation operator which
absorbs frame and spin dyad gauge terms. This formalism is applicable
to both the standard spacetime (i.e. $SL(2,\mathbb{C})$) 2-spinors as
well as to space (i.e. $SU(2,\mathbb{C})$) 2-spinors. We compute
expressions for the variations of the connection and the curvature spinors.
\end{abstract}

\section{Introduction}

Variational ideas play an important role in various areas of
mathematical General Relativity ---e.g. in the ADM formalism \cite{ArnDesMis62}, in the
analysis of Penrose-like inequalities \cite{Mar09} or in the analysis of area-angular
momentum inequalities \cite{Dai12} to mention some. Similarly, spinorial methods
constitute a powerful tool for the analysis and manipulation of the Einstein field
equations and their solutions ---most notably the proof of the
positivity of the mass by Witten \cite{Wit81} and the analysis of linearised
gravity, see e.g. \cite{PenRin84}. 

To the best of our knowledge, all available treatments of calculus of
variations and linerisations in spinorial settings make use of
computations in terms of components with respect to a dyad. It is
therefore of interest to have a setup for performing a
dyad-independent calculus of variations and computation of linearisations with spinors. The purpose of
the present article is to develop such a setup. We expect this
formalism to be of great value in both the analysis of
the notion of non-Kerrness introduced in \cite{BaeVal10a,BaeVal10b}
and positivity of the mass in \cite{BaeVal11a}, as well as in a covariant
analysis of linearised gravity.

The transformation properties of tensors and spinors pose some
conceptual subtleties which have to be taken into account when
computing variations of the basic tensorial and spinorial structures. It
is possible to have variations of of these structures which are \emph{pure
gauge}. This difficulty is usually dealt with by a careful fixing of the gauge
in some geometrically convenient manner. One thus makes calculus of variations in a specific gauge
and has to be careful in distinguishing between properties which are specific
to the particular gauge and those which are generic. This situation
becomes even more complicated as, in principle, both the tensorial and
spinorial structures are allowed to vary simultaneously.

\usemedskip
In this article it is shown that it is possible to define a
\emph{modified variation operator} which absorbs gauge terms in the
variation of spinorial fields and thus, allows to perform
\emph{covariant variations}. The idea behind this modified variation
operator is similar to that behind the derivative operators in
the GHP formalism which absorb terms associated to the freedom in a NP tetrad
---see \cite{GerHelPen73}. As a result of our analysis we are able to
obtain expressions involving abstract tensors and spinors ---thus,
they are valid in any system of coordinates, and therefore invariant under diffeomorphisms 
which are constant with respect to variations. However, linearisations of 
diffeomorphisms do affect our variational quantities. This is discussed in 
Section~\ref{sec:diffeomorphisms}, where we also find that the diffeomorphism 
freedom can be controled by a gauge source function.

\usemedskip
Finally, we point out that although our primary concern in this
article is the construction of a formalism for the calculus of
variations of expressions involving spinors in a 4-dimensional
Lorentzian manifold, the methods can be adapted to a space-spinor formalism on
3-dimensional Riemannian manifolds. This is briefly discussed in Section~\ref{sec:spacespinors}.

\usemedskip
The calculations in this article have been carried out in the Mathematica based symbolic differential
geometry suite \emph{xAct} \cite{xAct}, in particular \emph{SymManipulator} \cite{Bae11a} developed by TB.

\subsection*{Notation and conventions}
All throughout, we use abstract index notation to denote
tensors and spinors. In particular, the indices $a, b, c,\ldots$ and $i, j, k, \ldots$ are
abstract spacetime and spatial tensor indices respectively, while $A, B, C,\ldots$ denote 
abstract spinorial indices. 
The boldface indices $\bfa, \bfb, \bfc,\ldots$ and $\bfA, \bfB, \bfC,\ldots$ will be used as tensor
frame indices and spinor frame indices, respectively. We follow the 
tensorial and spinorial conventions of Penrose \& Rindler
\cite{PenRin84}. 

\usemedskip
Our signature convention for 4-dimensional Lorentzian metrics is
$(+,-,-,-)$, and 3-dimen\-sional Riemannian metrics have
signature $(-,-,-)$.

\usemedskip
The standard positions for the basic variations are $\delta g_{ab}$, $\delta
\sigma_{a}{}^{AA'}$, $\delta
\sigma_{k}{}^{AB}$, $\delta \omega^{\bf a}{}_{b}$ , $\delta
\epsilon^{\bf A}{}_{B}$, $\delta \epsilon_{AB}$, $\delta\gamma_a{}^b{}_c$. If any other index
positions appear, this means that the indices are moved up or down
with $g_{ab}$ or $\epsilon_{AB}$ after the variation. The definitions
of the above objects will be given in the main text.

\section{Basic setup}
In this section we discuss our basic geometric setup, which will be used in Section~\ref{Section:SLCalculus} to perform calculus of variations. 

\subsection{Families of metrics}
In what follows, let $(\mathcal{M},\mathring{g}_{ab})$ denote a
4-dimensional Lorentzian manifold (\emph{spacetime}). The metric
$\mathring{g}_{ab}$ will be known as the \emph{background metric}. In
what follows, in addition to $\mathring{g}_{ab}$, we consider \emph{arbitrary} families of
Lorentzian metrics $\{g_{ab}[\lambda]\}$ over $\mathcal{M}$ with
$\lambda\in\mathbb{R}$ a parameter such that
$g_{ab}[0]=\mathring{g}_{ab}$. Intuitively, a particular choice of
family of metrics can be thought of as a curve in the moduli space of
Lorentzian metrics over $\mathcal{M}$. The fact that we allow for
arbitrary families of metrics enables us to probe all possible directions of
this space in a neighbourhood of $\mathring{g}_{ab}$ and thus, 
we can compute \emph{Fr\'echet} derivatives of functionals
depending on the metric ---see Section~\ref{Section:CVBasics}. 

In order to make possible the discussion of spinors, it will be assumed
that the spacetimes $(\mathcal{M},g_{ab}[\lambda])$ for fixed $\lambda$ are
orientable and time orientable and admit a spinorial structure.

\medskip
\noindent
\textbf{Notational warning.} In what follows, for the ease of the presentation,
we often suppress the dependence on $\lambda$ from the various
objects. Thus, unless otherwise stated, all objects not
tagged with a \emph{ring} $(\mathring{\phantom{X}})$ are
assumed to depend on a parameter $\lambda$. 

\subsection{Frames}
In what follows, we assume that associated to each family of metrics
$\{ g_{ab}\}$ one has a family $\{ e_{\bf a }{}^a\}$ of
$g_{ab}$-orthonormal frames.  Let $\{ \omega^{\bf
  a}{}_a \}$ denote the family of associated cobases so that for fixed
$\lambda$ one has  $e_{\bf a}{}^a
\omega^{\bf b}{}_a =\delta_{\bf a}{}^{\bf b}$. Following the conventions of
the previous section, we write $\mathring{e}{}_{\bf a}{}^a\equiv 
e_{\bf a}{}^a[0]$ and $\mathring{\omega}{}^{\bf a}{}_a \equiv 
\omega^{\bf a}{}_a[0]$. By assumption, one has that
\begin{equation}
g_{ab} e_\bfa{}^a e_\bfb{}^b =\eta_{\bfa\bfb}, \qquad g_{ab} =
\eta_{\bfa\bfb} \omega^\bfa{}_a \omega^\bfb{}_b.
\label{OrthonormalityMetricFrame}
\end{equation}
where, as usual, $\eta_{\bfa\bfb} =\mbox{diag}(1,-1,-1,-1)$. 

\begin{remark}
Observe that in view of the relations
\eqref{OrthonormalityMetricFrame} any family of frames and coframes
$\{e'_\bfa{}^a\}$ and $\{ \omega'^{\bfa}{}_a \}$ related to
$\{e_\bfa{}^a\}$ and $\{ \omega^{\bfa}{}_a \}$ through a family of
Lorentz transformations $\{\Lambda^\bfa{}_\bfb\}$ give rise to the the
same family of metrics $\{ g_{ab} \}$ ---see Appendix \ref{Section:LorentzTransformations}. 
\end{remark}

 \subsection{Spinors}
By assumption, the spacetimes $(\mathcal{M},g_{ab})$ are endowed with
a spinorial structure. Accordingly, we consider families of
antisymmetric spinors $\{ \epsilon_{AB}\}$ such that for fixed
$\lambda$ the spinor $\epsilon_{AB}$ gives rise to the spinor
structure of $(\mathcal{M},g_{ab})$. Moreover, we set
$\mathring{\epsilon}_{AB} \equiv \epsilon_{AB}[0]$. 

\usemedskip
Associated to the family $\{ \epsilon_{AB}\}$ one considers a family
$\{ \epsilon_{\bf A}{}^A \}$ of normalised spin dyads 
---that is, one has that 
\begin{equation}
\epsilon_{AB}  \epsilon_{\bf A}{}^A \epsilon_{\bf B}{}^B
=\epsilon_{\bf AB}, \qquad 
\epsilon_{\bf AB} \equiv \left(  
\begin{array}{cc}
0 & 1 \\
-1 & 0
\end{array}
\right). \label{Definition:EpsilonSpinor}
\end{equation}
Let $\{ \epsilon^{\bfA}{}_A \}$ denote the family of dual \emph{covariant bases} 
for which the relation $\epsilon^{AB} \epsilon^{\bfA}{}_A 
\epsilon^{\bf B}{}_B=\epsilon^{\bf AB}$ 
with $(\epsilon^{\bf AB}) \equiv -(\epsilon_{\bf AB})^{-1}$ holds. 
It follows that one has
\[
\delta_A{}^B= \epsilon_A{}^{\bf A} \epsilon_{\bf A}{}^B, \qquad
\epsilon_{AB} = \epsilon_{\bf AB} \epsilon_A{}^{\bf A}
\epsilon_B{}^{\bf B}, \qquad \epsilon^{AB} = \epsilon^{\bf
  AB}\epsilon_{\bf A}{}^A \epsilon_{\bf B}{}^B.
\]

\begin{remark}
As in the case of tensor frames, any family of 
dyads $\{\epsilon'_{\bf A}{}^A \}$ related to $\{ \epsilon_{\bf A}{}^A
\}$ through a family of Lorentz transformations $\{\Lambda^\bfA{}_\bfB\}$ gives rise
to the same spinorial structures associated to the family of
antisymmetric spinors $\{ \epsilon_{AB} \}$ ---see Appendix \ref{Section:LorentzTransformations}. 
\end{remark}

\subsection{Infeld-van der Waerden and soldering forms}
The well-known correspondence between tensors and spinors is realised by the 
\emph{Infeld-van der Waerden symbols} $\sigma_{\bf a}{}^{\bf AA'}$ and 
$\sigma^{\bf a}{}_{\bf AA'}$. Given an arbitrary $ v^a \in T \mathcal{M}$
and $\beta_a \in T^* \mathcal{M}$ one has that
\[
v^{\bf a} \mapsto v^{\bf AA'} = v^{\bf a} \sigma_{\bf a}{}^{\bf AA'},
\qquad \beta_{\bf a} \mapsto \beta_{\bf AA'} =\beta_{\bf a}\sigma^{\bf a}{}_{\bf AA'} 
\]
where for \emph{fixed} $\lambda$ 
\[
v^{\bf a} \equiv v^a \omega^{\bf a}{}_a, \qquad \beta_{\bf a} \equiv
\beta_a e_{\bf a}{}^a, 
\]
denote the components of $v^a$ and $\beta_a$ with respect to the
orthonormal basis $e_\bfa{}^a[\lambda]$ of 
$(\mathcal{M},g_{ab}[\lambda])$. In more explicit terms, the
correspondence can be written as
\[
(v^{\bf 0}, v^{\bf 1}, v^{\bf 2}, v^{\bf 3}) \mapsto
\frac{1}{\sqrt{2}}
\left(
\begin{array}{cc}
v^{\bf 0} + v^{\bf 3} & v^{\bf 1} + \mbox{i} v^{\bf 2} \\
v^{\bf 1} - \mbox{i} v^{\bf 2} & v^{\bf 0} - v^{\bf 3}
\end{array}
\right), \;
(\beta_{\bf 0}, \beta_{\bf 1}, \beta_{\bf 2}, \beta_{\bf 3})
\mapsto \frac{1}{\sqrt{2}}
\left(
\begin{array}{cc}
\beta_{\bf 0} + \beta_{\bf 3} & \beta_{\bf 1} - \mbox{i}\beta_{\bf 2}
\\
\beta_{\bf 1} + \mbox{i}\beta_{\bf 2} & \beta_{\bf 0} - \beta_{\bf 3}
\end{array}
\right).
\]

From the \emph{Infeld-van der Waerden symbols} we define the 
\emph{soldering form} $\sigma_a{}^{AA'}$ and the dual of the soldering form $\sigma^a{}_{AA'}$ by
\begin{subequations}
\begin{align}
\sigma_a{}^{AA'} \equiv \epsilon_{\bf A}{}^A \bar{\epsilon}_{\bf
  A'}{}^{A'} \omega^{\bf  a}{}_a\sigma_{\bf a}{}^{\bf AA'},\label{Definition:SolderingForm}\\
\sigma^a{}_{AA'} \equiv \epsilon^{\bf A}{}_A \bar{\epsilon}^{\bf
  A'}{}_{A'} e_{\bf a}{}^a\sigma^{\bf a}{}_{\bf AA'} .\label{Definition:DualSolderingForm}
\end{align}
\end{subequations}

By direct calculation, we can then verify the relations
\begin{subequations}
\begin{align}
g_{ab} ={}&\epsilon_{AB} \bar\epsilon_{A'B'} \sigma_a{}^{AA'} \sigma_b{}^{BB'},\label{eq:gepsilonrel}\\
\delta_{a}{}^{b} ={}& \sigma_a{}^{BB'} \sigma^b{}_{BB'}. \label{eq:dualsolderingform}
\end{align}
\end{subequations}
It is important to note that $\sigma_a{}^{AA'}$ and $\sigma^a{}_{AA'}$ are tensor frame and spin dyad dependent, while the relations \eqref{eq:gepsilonrel} and \eqref{eq:dualsolderingform} are universal.

Following our approach, in the sequel we consider families
$\{\sigma_a{}^{AA'}\}$ and $\{ \sigma^a{}_{AA'} \}$ of soldering forms
such that $\mathring{\sigma}_a{}^{AA'}\equiv \sigma_a{}^{AA'}[0]$ and
$\mathring{\sigma}^a{}_{AA'}\equiv \sigma^a{}_{AA'}[0]$ are the
soldering forms associated to
$(\mathring{\omega}^{\bf b}{}_{a},\mathring{\epsilon}^{\bf B}{}_{A})$.

\begin{remark}
In this article we adopt the point of view that the
metric structure provided by $g_{ab}$ and the spinorial structure
given by $\epsilon_{AB}$ are independent from each other. After
a  choice of frame and spinor basis these structures are linked to
each other ---in an, admittedly, arbitrary manner--- through the relations in
\eqref{Definition:SolderingForm} and \eqref{eq:gepsilonrel}. 
\end{remark}

\section{Calculus of variations}
\label{Section:SLCalculus}

\subsection{Basic formalism}
\label{Section:CVBasics}

The main objective of our calculus of variations is to describe how real
valued functionals depend on their arguments  ---in particular, in the
case the arguments are covariant spinors. To motivate our analysis, we first consider
a real valued functional $\mathcal{F}[\omega^{\bf a}{}_a, \xi^{\bf
a}]$, where $\xi^{a}$ is a vector field and $\xi^{\bf a}=\omega^{\bf
a}{}_a\xi^a$.  Given a \emph{particular} family of fields $\{\omega^{\bf a}{}_a[\lambda],\, \xi^{\bf
a}[\lambda] \}$ depending on a parameter $\lambda$, we define the variations
$\{ \delta \omega^{\bf a}{}_a,\,  \delta \xi^{\bf a} \}$ through the expressions
\[
\delta \omega^{\bf a}{}_a \equiv \frac{\mbox{d}\omega^{\bf a}{}_a}{\mbox{d}\lambda}
\bigg|_{\lambda=0}, \qquad \delta \xi^{\bf a}
\equiv \frac{\mbox{d}\xi^{\bf a}}{\mbox{d}\lambda}  \bigg|_{\lambda=0}.
\]
In terms of the above fields over $\mathcal{M}$ we define \emph{the G\^ateux
derivative of $\mathcal{F}[\omega^{\bf a}{}_a, \xi^{\bf a}]$ at $\{\mathring{\omega}^{\bf a}{}_a,\, \mathring{\xi}^{\bf
a} \}$  in the direction of the family} $\{\omega^{\bf a}{}_a[\lambda],\, \xi^{\bf
a}[\lambda] \}$ as 
\begin{align*}
\delta_{\{\omega^{\bf a}{}_a,\, \xi^{\bf
a} \}} \mathcal{F}[\mathring\omega^{\bf a}{}_a, \mathring\xi^{\bf a}] 
\equiv{}& \frac{\mbox{d}}{\mbox{d}\lambda}
\mathcal{F}[\omega^{\bf a}{}_a[\lambda], \xi^{\bf a}[\lambda]] \bigg|_{\lambda=0}\\
={}& \frac{\mbox{d}}{\mbox{d}\lambda}
\mathcal{F}[\mathring{\omega}^{\bf a}{}_a + \lambda \delta \omega^{\bf a}{}_a, \mathring{\xi}^{\bf a}
+ \lambda  \delta \xi^{\bf a}] \bigg|_{\lambda=0}.
\end{align*}
Now, if $\delta_{\{\omega^{\bf a}{}_a,\, \xi^{\bf
a} \}} \mathcal{F}$ exists for
\emph{any} choice of family $\{\omega^{\bf a}{}_a,\, \xi^{\bf a} \}$ one then
says that $\mathcal{F}[\omega^{\bf a}{}_a, \xi^{\bf
a}]$ is \emph{Fr\'echet differentiable} at $\{\mathring{\omega}^{\bf a}{}_a,\, \mathring{\xi}^{\bf
a} \}$. If this is the case, there exists a functional $\delta
\mathcal{F}$, the \emph{Fr\'echet
  derivative}, from which 
$\delta_{\{\omega^{\bf a}{}_a,\, \xi^{\bf a} \}} \mathcal{F}$ can be computed if a particular choice of
the family of the variations $\{\delta \omega^{\bf a}{}_a,\, \delta \xi^{\bf
a} \}$ is considered. For more details concerning the notions of G\^ateaux and Fr\'echet
derivative and their relation see \cite{Tro83}.

The functional $\mathcal{F}[\omega^{\bf a}{}_a, \xi^{\bf a}]$
considered in the previous paragraph depends on the
coframe and components of a tensor fields in terms of this basis.  As
the particular choice of frame involves the specification of a gauge,
instead of regarding the functional $\delta \mathcal{F}$ as depending on the fields $[\mathring{\omega}^{\bf
a}{}_a,\mathring{\xi}^{\bf a},\delta\omega^{\bf a}{}_a, \delta\xi^{\bf
a}]$ it will be convenient to regard it as depending on 
$[\mathring{g}_{ab}, \mathring{\xi}^a, \delta
g_{ab},T_{ab},\delta\xi^{a}]$, where the field $T_{ab}$ describes the
frame gauge choice and
\[
\delta g_{ab} \equiv \frac{\mbox{d}g_{ab}}{\mbox{d}\lambda} \bigg|_{\lambda=0},
\]
where $\{ g_{ab} \}$ is a family of metrics over $\mathcal{M}$ such
that for fixed $\lambda$ the coframe $\mathring{\omega}^{\bf
a}{}_a$ is $g_{ab}$-orthonormal.

\usemedskip
Next, we consider real valued functionals depending on spinors. For
concreteness consider the a functional of the form $\mathcal{F}[g_{ab}, \epsilon^{\bf A}{}_A, \kappa_{\bf
A}]$. The G\^ateaux and Fr\'echet derivatives of this functional are
defined in the natural way by considering arbitrary families of fields
$\{g_{ab}, \, \epsilon^{\bf A}{}_A, \,\kappa_{\bf A}\}$ depending on a parameter $\lambda$. The variations implied by
this family of fields is then defined by 
\[
\delta g_{ab} \equiv \frac{\mbox{d}g_{ab}}{\mbox{d}\lambda} 
\bigg|_{\lambda=0}, 
 \qquad \delta \epsilon^{\bf A}{}_A \equiv \frac{\mbox{d}\epsilon^{\bf A}{}_A}{\mbox{d}\lambda} 
 \bigg|_{\lambda=0}, \qquad \delta \kappa_{\bf A} \equiv
 \frac{\mbox{d}\kappa_{\bf A}}{\mbox{d}\lambda}  \bigg|_{\lambda=0}.
\] 
In analogy to the example considered in the previous paragraphs, it
will be convenient to regard the Fr\'echet derivative 
$\delta \mathcal{F}$, which in principle depends on 
$[\mathring{g}_{ab}, \mathring{\epsilon}^{\bf A}{}_A, \mathring{\kappa}_{\bf A},
\delta g_{ab}, \delta\epsilon^{\bf A}{}_A, \delta \kappa_{\bf A}]$, 
as a functional of the arguments $[\mathring{g}_{ab}, \mathring{\epsilon}_{AB},
\mathring{\kappa}_A, \delta g_{ab}, \delta \kappa_A, T_{ab}, S_{AB}]$ where
the field $S_{AB}$ describes the dyad gauge choice. In this way one obtains a formalism 
that separates the tensor frame and spin dyad gauge in the Fr\'echet derivatives. The main 
observation in the sequel is that is is possible to obtain a 
\emph{modified variation operator} $\vartheta$ which absorbs the frame and dyad gauge terms 
so that the Fr\'echet derivative depends on the parameters $[\mathring{g}_{ab}, 
\mathring{\kappa}_A, \delta g_{ab}, \vartheta \kappa_A]$.

\medskip
\noindent
\textbf{Notational warning.} In what follows, for ease of
presentation, we mostly suppress the ring $\mathring{\phantom{x}}$
from the background quantities appearing in expressions involving
variations. If an expression does not involves variations then it
holds for both the background quantities and any other one in the
family. 

\subsection{Basic formulae for frames}
\label{Subsection:BasicVariationFrames}
Consider first the expression for the metric $g_{ab}$ in terms of the
coframe $\{ \omega^{\bf a}{}_a \}$ ---namely
\[
g_{ab} = \eta_{\bf a b} \omega^{\bf a}{}_a \omega^{\bf b}{}_b.
\]
Applying the variational operator $\delta$ to the above expression, using the Leibnitz rule, and that $\eta_{\bf a b}$ are
constants, yields
\begin{equation}
\delta g_{ab} = \eta_{\bf ab} \delta \omega^{\bf a}{}_a \omega^{\bf
  b}{}_b + \eta_{\bf ab}\omega^{\bf a}{}_a \delta
\omega^{\bf b}{}_b.
\label{VariationMetric}
\end{equation}
In certain computations it is useful to be able to express $\delta
\omega^{\bf a}{}_a$ in terms of $\delta g_{ab}$. In order to do this,
it is noticed that from \eqref{VariationMetric} it follows that
\[
\delta g_{ab} = 2 \eta_{\bf ab} \omega^{\bf a}{}_a \delta \omega^{\bf
  b}{}_b - 2 T_{ab},
\]
where 
\[
T_{ab} \equiv  \eta_ {{\bf c} {\bf d}}\omega^{\bf d}{}_{[a} \delta \omega^{{\bf c}}{}_{b]}.
\]
It then follows that 
\begin{equation}
\delta (\omega^{{\bf a}}{}_{a}) =  \tfrac{1}{2} e_{{\bf b}}{}^{b}
\eta^{{\bf b} {\bf a}} \delta g_{ab} - e_{{\bf b}}{}^{b} \eta^{{\bf b}
  {\bf a}} T_{ab}.
\label{VariationCoframeInTermsVariationMetric}
\end{equation}

A formula for the variation of the inverse metric can be 
computed by taking variations of the defining relation $\delta_a{}^b =
g_{ac} g^{cb}$. One finds that
\[
\delta(g^{dc}) ={}- g^{ad} g^{bc} \delta g_{ab}.
\]

A formula for the variation of the frame vectors $\{
e_{\bf a}{}^a\}$ in terms of the variation of $\delta \omega^{{\bf
    c}}{}_{b}$ is obtained by computing the variation of the expression
$\delta_{\bf a}{}^{\bf b} = e_{\bf a}{}^a \omega^{\bf b}{}_a$. One
finds that 
\[
\delta(e_{{\bf a}}{}^{d}) ={}- e_{{\bf a}}{}^{b} e_{{\bf c}}{}^{d} \delta \omega^{{\bf c}}{}_{b}.
\]

\usemedskip
The previous expressions can be used to compute a formula for the
variation of a covector $\xi_a$. Writing $\xi_a = \xi_{\bf a} \omega^{\bf
  a}{}_a$, one obtains that 
\begin{align*}
\delta\xi_{a}={}&\omega^{\bf b}{}_{a} \delta(\xi_{{\bf b}})
 + \tfrac{1}{2} e_{{\bf c}}{}^{b} \eta^{{\bf c} {\bf d}}  \xi_{{\bf d}}\delta g_{ab}
 -  e_{{\bf c}}{}^{b} \eta^{{\bf c} {\bf d}}  \xi_{{\bf d}}T_{ab}.
\end{align*}

\begin{remark}
An interpretation of the tensor $T_{ab}$ appearing in equation
\eqref{VariationCoframeInTermsVariationMetric} can be obtained by
considering a situation where $\delta g_{ab}=0$. In that case equation
\eqref{VariationCoframeInTermsVariationMetric} reduces to 
\[
\delta \omega^{{\bf a}}{}_{a} =  - e_{{\bf b}}{}^{b} \eta^{{\bf b}
  {\bf a}} T_{ab}.
\]
Writing $T_{ab} = T_{\bf ab} \omega^{\bf a}{}_a \omega^{\bf b}{}_b$
where $T_{\bf ab}$ denote the components of $T_{ab}$ with respect to
the coframe $\{ \omega^{\bf a}{}_a \}$ one  has that
\begin{equation*}
\delta \omega^{\bf a}{}_a = - e_{\bf b}{}^b \eta^{\bf ba} T_{\bf cd}
\omega^{\bf c}{}_a \omega^{\bf d}{}_b 
 = T^{\bf a}{}_{\bf c}
\omega^{\bf c}{}_a,
\end{equation*}
where $T^{\bf a}{}_{\bf c} \equiv - \eta^{\bf da} T_{\bf cd}
\omega^{\bf c}{}_a$.  Comparing with the discussion in Section
\ref{Section:LorentzTransformations} one sees that $T_{ab}$ encodes a rotation of
the basis. With this observation, in what follows we interpret the
second term in equation \eqref{VariationCoframeInTermsVariationMetric}
as a gauge term. 
\end{remark}

\subsection{Basic formulae for spinors}
\label{Subsection:BasicVariationSpinors}
The analysis in the previous section admits a straightforward
spinorial analogue. Given a covariant spinorial dyad $\{
\epsilon^{\bf A}{}_A\}$ one can  write
\[
\epsilon_{AB} = \epsilon_{\bf AB} \epsilon^{\bf A}{}_A \epsilon^{\bf B}{}_B.
\]
Thus, one has that
\begin{align*}
\delta \epsilon_{AB} ={}& 
\epsilon_{\bf AB} \epsilon^{\bf B}{}_B\delta \epsilon^{\bf A}{}_A 
+ \epsilon_{\bf AB} \epsilon^{\bf A}{}_A \delta\epsilon^{\bf B}{}_B \nonumber\\
={}& 2\epsilon_{\bf AB} \epsilon^{\bf  B}{}_B \delta \epsilon^{\bf A}{}_A - 2 S_{AB},
\end{align*}
where
\[
S_{AB} \equiv \epsilon_{\bf AB} \epsilon^{\bf B}{}_{(B} \delta
\epsilon^{\bf A}{}_{A)}.
\]
The variation of the contravariant antisymmetric spinor
$\epsilon^{AB}$ can be computed from the above formulae by
first computing the variation of
$\epsilon_{AB} \epsilon^{BC} = - \delta_{A}{}^{C}$ and then
multiplying with $\epsilon^{AD}$.  We  obtain that
\[
\delta(\epsilon^{DC}) ={} - \epsilon^{AD} \epsilon^{BC}\delta \epsilon_{AB}.
\]
As, $\delta\epsilon_{AB}$ is antisymmetric we can fully express it in terms of its trace as
$\delta\epsilon_{AB}=-\tfrac{1}{2}\epsilon_{AB}\delta\epsilon^C{}_{C}$.

Now, if one wants to compute $\delta \epsilon^{\bf A}{}_A$ in
terms of $\delta \epsilon_{AB} $ one has that
\begin{equation}
\delta \epsilon^{{\bf A}}{}_{A} =  \tfrac{1}{2} \epsilon^{{\bf A} {\bf
    B}}  \epsilon_{{\bf B}}{}^{B} \delta \epsilon_{AB} + \epsilon^{{\bf A}
  {\bf B}}  \epsilon_{{\bf B}}{}^{B} S_{AB}.
\label{VariationSpinBasisInTermsVariationEpsilon}
\end{equation}
If we compute the variation of $\delta_{{\bf A}}{}^{{\bf C}} = \epsilon^{\bf C}{}_{B} \epsilon_{{\bf A}}{}^{B}$ and multiply with $\epsilon_{{\bf C}}{}^{D}$ we get 
\[
\delta (\epsilon_{{\bf A}}{}^{A}) = - \epsilon_{{\bf A}}{}^{B} \epsilon_{{\bf C}}{}^{A}
\delta \epsilon^{{\bf C}}{}_{B}.
\]

\usemedskip
Now consider a covariant spinor $\phi_A$ and expand it with respect to
the spinor dyad $\{  \epsilon^{\bf A}{}_A\}$ as 
\[
\phi_A = \phi_{\bf A} \epsilon^{\bf A}{}_A. 
\]
A calculation using equation
\eqref{VariationSpinBasisInTermsVariationEpsilon} yields the
expression
\begin{eqnarray*}
&& \delta \phi_A = \delta \phi_{\bf A} \epsilon^{\bf A}{}_A +
\phi_{\bf A}
\delta \epsilon^{\bf A}{}_A \\
&& \phantom{\delta \phi_A} = \delta \phi_{\bf A} \epsilon^{\bf A}{}_A +
\tfrac{1}{2}\phi_{\bf A} \epsilon^{{\bf A} {\bf P} } \epsilon_{\bf P} {}^B \delta
\epsilon_{AB} + \phi_{\bf A} \epsilon^{{\bf A}{\bf P} } \epsilon_{\bf P}{}^B S_{AB}.
\end{eqnarray*}
Using the identity $\epsilon^{{\bf A}{\bf C} }\phi_{\bf A}
\epsilon_{\bf C}{}^B = \epsilon^{CB}\phi_C$ the variation $\delta
\phi_A$ can be reexpressed as
\[
\delta \phi_A = (\delta \phi_{\bf A}) \epsilon^{\bf A}{}_A +
\tfrac{1}{4}(\delta \epsilon^Q{}_Q) \phi_A - S_A{}^B \phi_B.
\]

\begin{remark}
As in the case of equation
\eqref{VariationCoframeInTermsVariationMetric} and the tensor
$T_{ab}$, the spinor $S_{AB}$ admits the interpretation of a
rotation. Indeed, considering a situation where $\delta
\epsilon_{AB}=0$, writing $S_{AB} = \epsilon^{\bf A}{}_A \epsilon^{\bf
  B}{}_B S_{\bf AB}$ one  finds that 
\begin{eqnarray*}
&& \delta \epsilon^{\bf A}{}_A = \epsilon^{\bf AB} \epsilon_{\bf
  B}{}^B S_{AB}\\
&& \phantom{\delta \epsilon^{\bf A}{}_A} = \epsilon^{\bf AB}
\epsilon_{\bf B}{}^B \epsilon^{\bf P}{}_A \epsilon^{\bf Q}{}_B S_{\bf
  PQ}\\
&& \phantom{\delta \epsilon^{\bf A}{}_A} = S^{\bf A}{}_{\bf B}
\epsilon^{\bf B}{}_A.
\end{eqnarray*}
Comparing with Appendix~\ref{Section:LorentzTransformations}, we find that $S_{AB}$ encodes a rotation of the spin
dyad. 
\end{remark}

\subsection{Variation of the soldering form}
In the reminder of this article we will consider a more general
setting in which both the metric $g_{ab}$ and the antisymmetric spinor
$\epsilon_{AB}$ can be varied simultaneously. 
To analyse the relation between
the variations of these two structures it is convenient to consider the
soldering form $\sigma{}_a{}^{AA'}$. 

To compute the variation of the soldering form, one starts by
computing the variation of the relation
\eqref{Definition:SolderingForm}. As we are treating the Infeld-van
der Waerden symbols as constants, their variation vanishes 
--- that is, although both the metric and spinor structure may vary,
the formal relation between tetrads and spin dyads will be preserved.
A direct combination of the methods of Sections
\ref{Subsection:BasicVariationFrames} and
\ref{Subsection:BasicVariationSpinors} on formula
\eqref{Definition:SolderingForm} lead, after a computation, to the expression
\begin{align}
\delta \sigma_a{}^{AA'}={}&  \tfrac{1}{2} \delta
\epsilon^{A}{}_{B} \sigma_a{}^{BA'}+\tfrac{1}{2} \delta\bar \epsilon^{A'}{}_{B'} \sigma_a{}^{AB'}
+ \tfrac{1}{2} g^{bc}\delta g_{ab}
\sigma_c{}^{AA'} \nonumber\\
&- \bar{S}^{A'}{}_{B'} \sigma_a{}^{AB'} -  S^{A}{}_{B}
\sigma_a{}^{BA'}  -  T_a{}^b \sigma_b{}^{AA'}.
\label{deltasigmaeq}
\end{align}
The terms in the second line of the previous expression are identified
as \emph{gauge terms}. Observe that in this case one has two types of
gauge terms: one arising from the variation of the tensor frame
and one coming from the variation of the spin frame. 

If we compute the variation of equation
\eqref{eq:dualsolderingform} and multiply with $\sigma^a{}_{AA'}$ we
get 
\[
\delta(\sigma^b{}_{AA'}) = - \delta(\sigma_a{}^{BB'})
\sigma^a{}_{AA'} \sigma^b{}_{BB'}.
\]
Multiplying equation \eqref{deltasigmaeq} with
$g^{ac}\sigma_c{}^{BB'}$ and splitting into irreducible parts, we get
the relations
\begin{align*}
\delta \sigma^a{}_{(A}{}^{(A'}\sigma_{|a|B)}{}^{B')} ={}&\tfrac{1}{2} \delta g_{(AB)}{}^{(A'B')},\\
\delta \sigma^{a(A|B'|}\sigma_a{}^{B)}{}_{B'}={}&T^{AB}
 - 2 S^{AB},\\
\delta \sigma^{aB(A'}\sigma_{|a|B}{}^{B')}={}&\bar{T}^{A'B'}
 - 2 \bar{S}^{A'B'},\\
\delta \sigma^{aBB'} \sigma_{aBB'}={}&\tfrac{1}{2} \delta g^{B}{}_{B}{}^{B'}{}_{B'}
 + \delta \epsilon^{B}{}_{B}
 + \delta \bar\epsilon^{B'}{}_{B'},
\end{align*}
where we have defined
\begin{align*}
T_{AB}\equiv{}& T_{ab} \sigma^a{}_{A}{}^{A'} \sigma^b{}_{BA'},&
\delta g_{ABA'B'}\equiv{}& \delta g_{ab} \sigma^a{}_{AA'} \sigma^b{}_{BB'}.
\end{align*}

\subsection{General variations of spinors}
The formulae for the variations of the soldering form and its dual can
now be used to compute the variation of arbitrary spinors under
variations of the metric and spinor structures. To this end, consider
spinors $\zeta^{AA'}$ and $\xi_{AA'}$. Making use of the Leibnitz rule
one obtains the expressions
\begin{subequations}
\begin{align}
 \sigma_{a}{}^{AA'}\delta\zeta^a={}&\delta(\zeta^{AA'})
 -  \tfrac{1}{4} \delta \epsilon^{B}{}_{B} \zeta^{AA'}
 -  \tfrac{1}{4} \delta \bar\epsilon^{B'}{}_{B'} \zeta^{AA'}  -  \tfrac{1}{2} \delta g^{A}{}_{B}{}^{A'}{}_{B'} \zeta^{BB'}
 \nonumber\\
&-  \tfrac{1}{2} \bar{T}^{A'}{}_{B'} \zeta^{AB'}
 + \bar{S}^{A'}{}_{B'} \zeta^{AB'}-  \tfrac{1}{2} T^{A}{}_{B} \zeta^{BA'}
  + S^{A}{}_{B} \zeta^{BA'}
, \label{GeneralVariation1}\\
\sigma^{a}{}_{AA'}\delta\xi_a ={}&\delta(\xi_{AA'})
 + \tfrac{1}{4} \delta \epsilon^{B}{}_{B} \xi_{AA'}
 + \tfrac{1}{4} \delta \bar\epsilon^{B'}{}_{B'} \xi_{AA'}+ \tfrac{1}{2} \delta g_{A}{}^{B}{}_{A'}{}^{B'} \xi_{BB'}\nonumber\\
&+  \tfrac{1}{2} \bar{T}_{A'}{}^{B'} \xi_{AB'}
 - \bar{S}_{A'}{}^{B'} \xi_{AB'}+  \tfrac{1}{2} T_{A}{}^{B} \xi_{BA'}
  - S_{A}{}^{B} \xi_{BA'}, \label{GeneralVariation2}
\end{align}
\end{subequations}
where $\zeta^a\equiv\sigma^{a}{}_{BB'}\zeta^{BB'}$ and $\xi_a\equiv\sigma_{a}{}^{BB'}\xi_{BB'}$.
We observe that both expressions contain a combination of \emph{gauge
  terms} involving the spinors $T_{AB}$ and $S_{AB}$.

\usemedskip
In view of the discussion in the previous paragraph we introduce a
general \emph{modified variation operator}. 
\begin{definition}\label{def:modvar1}
The \emph{modified variation operator} $\vartheta$ is for valence 1 spinors defined by
\begin{align*}
\vartheta\phi_{A}\equiv{}&\delta\phi_{A}
 + \tfrac{1}{4} \delta \epsilon^{B}{}_{B} \phi_{A}
 +  \tfrac{1}{2} T_{A}{}^{B} \phi_{B}
 - S_{A}{}^{B} \phi_{B},\\
\vartheta\phi^{A}\equiv{}&\delta\phi^{A}
 -  \tfrac{1}{4} \delta \epsilon^{B}{}_{B} \phi^{A}
 -  \tfrac{1}{2} T^{A}{}_{B} \phi^{B}
 + S^{A}{}_{B} \phi^{B},\\
 \vartheta\bar{\phi}_{A'}\equiv{}&\delta\bar{\phi}_{A'}
 + \tfrac{1}{4} \delta \bar\epsilon^{B'}{}_{B'} \bar{\phi}_{A'}
 +  \tfrac{1}{2} \bar{T}_{A'}{}^{B'} \bar{\phi}_{B'}
 - \bar{S}_{A'}{}^{B'} \bar{\phi}_{B'},\\
\vartheta\bar{\phi}^{A'}\equiv{}&\delta\bar{\phi}^{A'}
 -  \tfrac{1}{4} \delta \bar\epsilon^{B'}{}_{B'} \bar{\phi}^{A'}
 -  \tfrac{1}{2} \bar{T}^{A'}{}_{B'} \bar{\phi}^{B'}
 + \bar{S}^{A'}{}_{B'} \bar{\phi}^{B'},
\end{align*}
and extended to arbitrary valence spinors by the Leibnitz rule. 
\end{definition}

In particular, using the above
definitions in expressions
\eqref{GeneralVariation1}-\eqref{GeneralVariation2} one finds that 
\begin{align*}
&\sigma_{a}{}^{AA'}\delta\zeta^{a} ={}\vartheta\zeta^{AA'}
 -  \tfrac{1}{2} \delta g^{A}{}_{B}{}^{A'}{}_{B'} \zeta^{BB'},\\
&\sigma^{a}{}_{AA'}\delta\xi_{a} ={}\vartheta\xi_{AA'}
 + \tfrac{1}{2} \delta g_{A}{}^{B}{}_{A'}{}^{B'} \xi_{BB'},
\end{align*}
showing that $\vartheta\zeta^{AA'}$ and $\vartheta\zeta_{AA'}$ are
frame gauge independent. Moreover, a further calculation shows that 
\[
\vartheta \epsilon_{AB}=0
\]
so that the process of raising and lowering spinor indices commutes
with the modified variation $\vartheta$ operator.

\begin{remark}
Expanding the $\phi_{A}$ in terms of the spin dyad in the $\delta\phi_{A}$ term in Definition~\ref{def:modvar1} gives
\begin{align}
\vartheta(\phi_{A})={}& \epsilon^{\bf B}{}_{A} \delta(\phi_{{\bf B}})   +  \tfrac{1}{2} T_{A}{}^{B} \phi_{B}.
\end{align}
Observe that the $S_{AB}$ and $\delta\epsilon_{AB}$ terms cancel out.
\end{remark}

\section{Variations and the covariant derivative}
\label{Section:CovDevs}

The purpose of this section is to to analyse te relation between the
variation operators $\delta$ and $\vartheta$ and the Levi-Civita
connection $\nabla_a$ of the metric $g_{ab}$. 

\subsection{Basic tensorial relations}

Our analysis of the variations of expressions involving covariant
derivatives is based on the following basic assumption:

\begin{assumption}
For any scalar field $f$ over $\mathcal{M}$ one
has that
\begin{equation}
\nabla_a\delta f=\delta(\nabla_a f)
\label{CommutationNabladelta}
\end{equation}
\end{assumption}

In what follows, define the \emph{frame dependent tensor}
\[
\gamma_a{}^b{}_c \equiv{} - e_{{\bf c}}{}^b \nabla_{a}\omega^{\bf c}{}_{c}.
\]
The tensor $\gamma_a{}^b{}_c$ can be regarded as a convenient way of
grouping the connection coefficients $\gamma_\bfa{}^\bfb{}_\bfc$ of
the connection $\nabla_a$ with respect to the frame $\{
e_\bfa{}^a\}$. A calculation shows, indeed, that 
\[
\gamma_a{}^b{}_c = \gamma_\bfa{}^\bfb{}_\bfc \,\omega^\bfa{}_a
\omega^\bfc{}_c e_\bfb{}^b.
\]
We can express all covariant derivatives of the cobasis and the basis in terms of $\gamma_a{}^b{}_c$ via
\[
\nabla_{a}\omega^{{\bf f}}{}_{c}={}- \omega^{\bf f}{}_{b}
\gamma_{a}{}^{b}{}_{c},\qquad 
\nabla_{d}e_{{\bf f}}{}^{b}={}e_{{\bf f}}{}^{a} \gamma_{d}{}^{b}{}_{a}.
\]
Differentiating the orthonormality condition  $\eta^{{\bf a} {\bf b}}
= \omega^{\bf a}{}_{c} \omega^{\bf b}{}_{d} g^{cd}$ 
and multiplying with $e_{{\bf a}}{}^{h} e_{{\bf b}}{}^{l}$ we 
get the relation
\begin{equation}
\gamma_{f}{}^{(a}{}_{c}g^{b)c}=0\label{eq:metriccompgamma}
\end{equation}
encoding the metric compatibility of $\nabla_a$. The variation of this gives
\begin{align}
\delta \gamma_{f}{}^{(ab)}={}&\gamma_{f}{}^{(a|c|}\delta g^{b)}{}_{c}.\label{eq:symdeltagamma}
\end{align}

Now, for any covector $\xi_a$, its covariant derivative can be
expanded in terms of the frame as 
\begin{align*}
\nabla_{a}\xi_{b}={}&- \omega^{\bf c}{}_{d} \gamma_{a}{}^{d}{}_{b} \xi_{{\bf c}}
 + \omega^{{\bf c}}{}_{b} \nabla_{a}\xi_{{\bf c}}.
\end{align*}
Computing the variation of this last expression, and using the
relations above, gives after some straightforward calculations
\begin{align}
\delta(\nabla_{a}\xi_{b})={}&- \delta \gamma_{a}{}^{c}{}_{b} \xi_{c}
 + T_{c}{}^{d} \gamma_{abd} \xi^{c}
 -  T_{b}{}^{d} \gamma_{acd} \xi^{c}
 + \tfrac{1}{2} \gamma_{ac}{}^{d} \delta g_{bd} \xi^{c}
 + \tfrac{1}{2} \gamma_{ab}{}^{d} \delta g_{cd} \xi^{c}
 + \xi^{c} \nabla_{a}T_{bc}\nonumber\\
& -  \tfrac{1}{2} \xi^{c} \nabla_{a}\delta g_{bc}
 + \nabla_{a}\delta \xi_{b}.\label{eq:deltanablaxi1}
\end{align}
In the previous calculation Assumption \eqref{CommutationNabladelta} has been used. If we
use relation \eqref{eq:deltanablaxi1} with $\xi_a=\nabla_a f$,
antisymmetrize over $a$ and $b$, and assume that the connection is
torsion free, we get
\begin{align*}
0={}& (T_{c}{}^{d} \gamma_{[ab]d}
 + \tfrac{1}{2} \delta g_{c}{}^{d} \gamma_{[ab]d}
 -  \delta \gamma_{[a|c|b]}
 + \nabla_{[a}T_{b]c}
 -  \tfrac{1}{2} \nabla_{[a}\delta g_{b]c}
 + T_{[a}{}^{d}\gamma_{b]cd}
 + \tfrac{1}{2} \gamma_{[a|c|}{}^{d}\delta g_{b]d})\nabla^{c}f.
\end{align*}
Hence, the torsion free condition is encoded by
\begin{align}
\delta \gamma_{[a|c|b]}={}&T_{c}{}^{d} \gamma_{[ab]d}
 + \tfrac{1}{2} \delta g_{c}{}^{d} \gamma_{[ab]d}
 + \nabla_{[a}T_{b]c}
 -  \tfrac{1}{2} \nabla_{[a}\delta g_{b]c}
 + T_{[a}{}^{d}\gamma_{b]cd}
 + \tfrac{1}{2} \gamma_{[a|c|}{}^{d}\delta g_{b]d}.\label{eq:torsioncond}
\end{align}
Now, using the identity
\begin{align*}
\delta \gamma_{abc}={}&\delta \gamma_{[a|b|c]}
 -  \delta \gamma_{[a|c|b]}
 + \delta \gamma_{[b|a|c]}
 + \delta \gamma_{a(bc)}
 -  \delta \gamma_{b(ac)}
 + \delta \gamma_{c(ab)},
\end{align*}
we can use equations \eqref{eq:symdeltagamma} and \eqref{eq:torsioncond} to compute
\begin{align}
\delta \gamma_{abc}={}&- T_{c}{}^{d} \gamma_{abd}
 + T_{b}{}^{d} \gamma_{acd}
 + \tfrac{1}{2} \gamma_{ac}{}^{d} \delta g_{bd}
 + \tfrac{1}{2} \gamma_{ab}{}^{d} \delta g_{cd}
 -  \nabla_{a}T_{bc}
 -  \tfrac{1}{2} \nabla_{b}\delta g_{ac}
 + \tfrac{1}{2} \nabla_{c}\delta g_{ab}.
\end{align}
It follows then that equation \eqref{eq:deltanablaxi1} can therefore be simplified to
\begin{equation}
\delta(\nabla_{a}\xi_{b})={}\nabla_{a}(\delta \xi_{b})
 -  \tfrac{1}{2} g^{cd} (\nabla_{a}\delta g_{bc}
 +  \nabla_{b}\delta g_{ac}
 -  \nabla_{c}\delta g_{ab}) \xi_{d} .\label{eq:varcovd1}
\end{equation}
It is important to observe that this formula is a tensorial
expression. Hence, it allows to define a \emph{transition tensor}
\begin{equation}
Q_b{}^a{}_c\equiv \tfrac{1}{2} g^{ad} (\nabla_{b}\delta g_{dc}
 +  \nabla_{c}\delta g_{bd}
 -  \nabla_{d}\delta g_{bc})
\label{Definition:TransitionTensor}
\end{equation}
 relating the connections $\nabla_a$ and $\delta \nabla_a$. This is
not surprising as it is well known that the space of covariant
derivatives on a manifold is an affine space. Making use the
definition of $Q_b{}^a{}_{c}$, equation \eqref{eq:varcovd1} takes the
suggestive form
\begin{equation}
\delta(\nabla_{a}\xi_{b})={}\nabla_{a}(\delta \xi_{b})
 - Q_b{}^d{}_a \xi_d.\label{eq:deltanablaxi}
\end{equation}
Furthermore, making use of the Leibnitz rule one finds that for an
arbitrary vector $v^a$ one has
\[
\delta(\nabla_a v^b) = \nabla_a (\delta v^b) + Q_c{}^b{}_a v^c.
\]
The extension to higher valence tensors follows in a similar manner.

\subsection{Spinorial expressions}
In order to discuss the variations of the spinor covariant derivative
$\nabla_{AA'}$ associated to the Levi-Civita connection $\nabla_a$ it
is convenient to define a spinorial analogue of the tensor
$\gamma_a{}^b{}_c$ ---namely
\[
\gamma{}_{a}{}^{B}{}_{C} \equiv{} - \epsilon_{{\bfC}}{}^{B} \nabla_{a}\epsilon^{\bfC}{}_{C}.
\]
The hybrid $\gamma{}_{a}{}^{B}{}_{C}$ is related to
$\gamma_{a}{}^{BB'}{}_{CC'}\equiv\gamma_a{}^b{}_c\sigma_b{}^{BB'}\sigma^c{}_{CC'}$
through the decomposition
\[
\gamma_{a}{}^{BB'}{}_{CC'} = \gamma_{a}{}^B{}_C \delta_{C'}{}^{B'} + \bar{\gamma}_{a}{}^{B'}{}_{C'} \delta_C{}^B. 
\]
It follows then that
\begin{equation}
\gamma_{a}{}^{B}{}_{C} = \tfrac{1}{2} \gamma_a{}^{c}{}_{b}\sigma^{b}{}_{CB'} \sigma_{c}{}^{BB'}.
\label{GammaTensorToGammaSpinor}
\end{equation}
From this last expression can then  be verified that
\[
\gamma_{aBC}= \gamma_{aCB}.
\]

\usemedskip
The variational derivative of $\gamma_{a}{}^{B}{}_{C}$ can be
computed using equation \eqref{GammaTensorToGammaSpinor}. One finds
that
\begin{align}
\delta(\gamma_a{}^{B}{}_{C})={}&\tfrac{1}{4} Q_{acd} \sigma^{dBB'} \sigma^{c}{}_{CB'}
 -  \tfrac{1}{4} Q_{acd} \sigma^{dBB'} \sigma^{c}{}_{CB'}
-  \tfrac{1}{2} \gamma_{acd} S_{CD} \sigma^{cBB'} \sigma^{dD}{}_{B'}\nonumber\\
&-  \tfrac{1}{2} \gamma_{acd} S^{B}{}_{D} \sigma^{c}{}_{C}{}^{B'} \sigma^{dD}{}_{B'}
 -  \tfrac{1}{2} \nabla_{a}T^{B}{}_{C}.
\label{VariationGammaSpinor}
\end{align}
In this last expression observe, in particular, the appearance of the
gauge spinors $S_{AB}$ and $T_{AB}$. In turn,
equation \eqref{VariationGammaSpinor} can be used to compute the
variation of the covariant derivative of an arbitrary spinor
$\kappa_A$.  Expanding $\kappa_A$ in terms of the spin dyad and
differentiating we get
\[
\nabla_{a}\kappa_{B} =  \epsilon^{\bf C}{}_{B} \nabla_{a}\kappa_{{\bf C}}- \gamma_a{}^{C}{}_{B} \kappa_{C} .
\]
It follows that the variation of this last expression is given by
\begin{align*}
\delta(\nabla_{a}\kappa_{A})={}&
  \nabla_{a}\delta \kappa_{A}
 -  \tfrac{1}{2} \kappa^{B} \nabla_{a}T_{AB}
 + \kappa^{B} \nabla_{a}S_{AB}
+ \tfrac{1}{4} \kappa_{A} \nabla_{a}\delta \epsilon^{B}{}_{B} \nonumber\\
&+\tfrac{1}{4} Q_{abc} \kappa^{B} \sigma^{b}{}_{A}{}^{A'} \sigma^{c}{}_{BA'}
 -  \tfrac{1}{4} Q_{acb} \kappa^{B} \sigma^{b}{}_{A}{}^{A'} \sigma^{c}{}_{BA'}\nonumber\\
={}& \nabla_{a}\vartheta \kappa_{A}
 -  \tfrac{1}{4} \delta \epsilon^{B}{}_{B} \nabla_{a}\kappa_{A}
  + \tfrac{1}{2} T_{AB} \nabla_{a}\kappa^{B}
 -  S_{AB} \nabla_{a}\kappa^{B}\nonumber\\
&+\tfrac{1}{4} Q_{abc} \kappa^{B} \sigma^{b}{}_{A}{}^{A'} \sigma^{c}{}_{BA'}
 -  \tfrac{1}{4} Q_{acb} \kappa^{B} \sigma^{b}{}_{A}{}^{A'} \sigma^{c}{}_{BA'}.
\end{align*}

\usemedskip
In order to write the spinorial derivative $\nabla_{AA'} \kappa_B$
(rather than $\nabla_a \kappa_B$) it is convenient to define the
spinor
\begin{equation}\label{Definition:QopSpinor}
\Qop_{AA'BC}\equiv - \tfrac{1}{2} \sigma^{a}{}_{AA'} \sigma^{b}{}_{B}{}^{B'} \sigma^{c}{}_{CB'} Q_{[bc]a}.
\end{equation}

\begin{theorem}
The variation of the covariant derivative of a spinor is given by
\begin{subequations}
\begin{align}\label{eq:ThmVarDerSpin}
\vartheta(\nabla_{AA'}\kappa_{B})={}&
  \nabla_{AA'}\vartheta \kappa_{B}
+\Qop_{AA'BC} \kappa^{C}
 -  \tfrac{1}{2} \delta g_{ACA'B'} \nabla^{CB'}\kappa_{B},\\
\vartheta(\nabla_{AA'}\bar{\kappa}_{B'})={}&
  \nabla_{AA'}\vartheta\bar{\kappa}_{B'}
+\bar\Qop_{A'AB'C'} \bar{\kappa}^{C'}
 -  \tfrac{1}{2} \delta g_{ABA'C'} \nabla^{BC'}\bar{\kappa}_{B'}.\label{eq:ThmVarDerSpinDg}
\end{align}
\end{subequations}
\end{theorem}
\begin{proof}
Using the expressions in the previous paragraphs one has that 
\begin{align*}
\delta(\nabla_{AA'}\kappa_{B})={}&
  \nabla_{AA'}\vartheta \kappa_{B}
+\Qop_{AA'BC} \kappa^{C}
 -  \tfrac{1}{4} \delta \epsilon^{C}{}_{C} \nabla_{AA'}\kappa_{B}
 + \tfrac{1}{2} T_{BC} \nabla_{AA'}\kappa^{C}
 -  S_{BC} \nabla_{AA'}\kappa^{C}\nonumber\\
& -  \delta \sigma_a{}^{CB'} \sigma^a{}_{AA'} \nabla_{CB'}\kappa_{B}\\
={}& \nabla_{AA'}\vartheta \kappa_{B}
+\Qop_{AA'BC} \kappa^{C}
 -  \tfrac{1}{2} \delta \epsilon^{C}{}_{C} \nabla_{AA'}\kappa_{B}
 -  \tfrac{1}{4} \delta \bar\epsilon^{B'}{}_{B'} \nabla_{AA'}\kappa_{B}
 + \tfrac{1}{2} T_{BC} \nabla_{AA'}\kappa^{C}\nonumber\\
& -  S_{BC} \nabla_{AA'}\kappa^{C}
 + \tfrac{1}{2} \bar{T}_{A'B'} \nabla_{A}{}^{B'}\kappa_{B}
 -  \bar{S}_{A'B'} \nabla_{A}{}^{B'}\kappa_{B}
 + \tfrac{1}{2} T_{AC} \nabla^{C}{}_{A'}\kappa_{B}\nonumber\\
& -  S_{AC} \nabla^{C}{}_{A'}\kappa_{B}
 -  \tfrac{1}{2} \delta g_{ACA'B'} \nabla^{CB'}\kappa_{B}.
\end{align*}
Expressing the above formula in terms of the modified variation $\vartheta$, we get \eqref{eq:ThmVarDerSpin}. The equation \eqref{eq:ThmVarDerSpinDg} is given by complex conjugation.
\end{proof}

\subsubsection{Decomposition of $\Qop_{AA'BC}$}
Starting from the definition in equation \eqref{Definition:QopSpinor},  a calculation yields 
\begin{align*}
\Qop_{AA'BC}={}&- \tfrac{1}{4} \sigma^{a}{}_{AA'} \sigma^{b}{}_{B}{}^{B'} \sigma^{c}{}_{CB'} \nabla_{b}\delta g_{ac}
 + \tfrac{1}{4} \sigma^{a}{}_{AA'} \sigma^{b}{}_{B}{}^{B'} \sigma^{c}{}_{CB'} \nabla_{c}\delta g_{ab}\\
={}&- \tfrac{1}{2} \nabla_{(B}{}^{B'}\delta g_{C)AB'A'}.
\end{align*}
The above expression can be conveniently decomposed in irreducible
terms. To this end, one defines
\[
G\equiv \delta g^{C}{}_{C}{}^{C'}{}_{C'},\qquad 
G_{ABA'B'} \equiv {}\delta g_{(AB)(A'B')}.
\]
If we also decompose $\Qop_{ABCA'}$ into irreducible parts, we get
\begin{align}\label{Decomposition:SpinorQop}
\Qop_{AA'BC}={}&- \tfrac{1}{2} \nabla_{(A}{}^{B'}G_{BC)A'B'}
 +  \tfrac{1}{8} \epsilon_{A(B}\nabla_{C)A'}G
 - \tfrac{1}{6} \epsilon_{A(B}\nabla^{DB'}G_{C)DA'B'}.
\end{align}

For future use we notice the following relations which follow from the
decomposition in irreducible components of equation
\eqref{Decomposition:SpinorQop} and the reality of $\delta g_{ABA'B'}$:
\begin{align*}
\Qop^{B}{}_{A'AB}={}&- \tfrac{3}{16} \nabla_{AA'}G
 + \tfrac{1}{4} \nabla_{BB'}G_{A}{}^{B}{}_{A'}{}^{B'},\\
\bar\Qop^{B'}{}_{AA'B'}={}&\Qop^{B}{}_{A'AB},\\
\nabla_{BB'}G_{CDA'}{}^{B'}={}&2 \Qop_{BA'CD}
 - 4 \Qop^{A}{}_{A'(C|A|}\epsilon_{D)B}
 -  \tfrac{1}{2} \epsilon_{(C|B|}\nabla_{D)A'}G,\\
 \nabla_{BA'}G_{A}{}^{B}{}_{B'C'}={}&2 \bar\Qop_{A'AB'C'}
 - 4 \Qop^{B}{}_{(B'|AB|}\bar\epsilon_{C')A'}
 -  \tfrac{1}{2} \bar\epsilon_{(B'|A'}\nabla_{A|C')}G.
\end{align*}

We also define the field
\begin{align}
F^{AA'}\equiv{}&\nabla_{BB'}\delta g^{ABA'B'} - \tfrac{1}{2} \nabla^{AA'}\delta g^{B}{}_{B}{}^{B'}{}_{B'}\label{eq:gaugesourcedef}\\
={}&\nabla_{BB'}G^{ABA'B'} - \tfrac{1}{4} \nabla^{AA'}G.\nonumber
\end{align}
In the next section we will see that this can be interpreted as a
gauge source function for the linearised diffeomorphisms.

\subsection{Diffeomorphism dependence}\label{sec:diffeomorphisms}
We will now briefly consider the dependence on diffeomorphisms.  Let
$\phi_\lambda$ be a one parameter group of diffeomorphisms generated
by a vector field $\xi^a$ and such that
$g_{ab}[\lambda]=\phi^*_{-\lambda}\mathring g_{ab}$. The metrics in
this family have the same geometric content 
and one readily finds that
\begin{align}
\delta g_{ab}={}&\mathcal{L}_\xi g_{ab}
=2 \nabla_{(a}\xi_{b)}.\label{eq:puregauge}
\end{align}
Moreover, a further computation yields
\begin{align*}
\Qop_{AA'BC}={}&- \tfrac{1}{2} \nabla_{(C}{}^{B'}\nabla_{B)B'}\xi_{AA'}
 -  \tfrac{1}{2} \nabla_{(C}{}^{B'}\nabla_{|AA'|}\xi_{B)B'},\\
F^{AA'}={}&\nabla_{BB'}\nabla^{BB'}\xi^{AA'}
-6 \Lambda \xi^{AA'}
 + 2 \Phi^{A}{}_{B}{}^{A'}{}_{B'} \xi^{BB'}.
\end{align*}

Given a general family of metrics $g_{ab}[\lambda]$, we can compute
the field $F^{AA'}$ associated to the family. Given any $\widetilde
F^{AA'}$, we can then solve the wave equation
\begin{align*} \widetilde F^{AA'} - F^{AA'}={}&-6 \Lambda \xi^{AA'} +
2 \Phi^{A}{}_{B}{}^{A'}{}_{B'} \xi^{BB'} +
\nabla_{BB'}\nabla^{BB'}\xi^{AA'}. 
\end{align*} The solution $\xi^a=\sigma^a_{AA'}\xi^{AA'}$ to this
equations will then give a one parameter group of diffeomorphisms
$\phi_\lambda$, such that
$g_{ab}[\lambda]=\phi^*_{-\lambda}g_{ab}[\lambda]$ has the same
geometric content, but with corresponding $\widetilde F^{AA'}$. With
this observation, we can interpret \eqref{eq:gaugesourcedef} as a
gauge source function for the linearised diffeomorphisms.

\section{Variation of curvature}
\label{Section:Curvature}
The purpose of this section is to compute the variation of the various
spinorial components of the curvature tensor. As it will be seen
below, the starting point of this computation is the commutator of
covariant derivatives. 

\usemedskip
We start by computing the variation of 
\[
\square_{(AB} \kappa_{C)} =\nabla_{(A}{}^{A'}\nabla_{B|A'|}\kappa_{C)}=-\Psi_{ABCD} \kappa^{D}
\]
for an arbitrary spinor $\kappa_A$. A direct calculation using the
Leibnitz rule for the modified commutator $\vartheta$ gives 
\begin{align*}
\Psi_{ABCD} \vartheta \kappa^{D} + \vartheta(\Psi_{ABCD}) \kappa^{D}={}&- \nabla_{(A}{}^{A'}\nabla_{B|A'|}\vartheta \kappa_{C)}
 + \tfrac{1}{4} G \nabla_{(A}{}^{A'}\nabla_{B|A'|}\kappa_{C)}\nonumber\\
& -  \tfrac{1}{2} G_{(A}{}^{DA'B'}\nabla_{B|A'}\nabla_{DB'|}\kappa_{C)}
 + \tfrac{1}{2} G_{(A}{}^{DA'B'}\nabla_{|DA'|}\nabla_{B|B'|}\kappa_{C)}\nonumber\\
& + \Qop_{(A}{}^{A'}{}_{B}{}^{D}\nabla_{|DA'|}\kappa_{C)}
 + \bar\Qop^{A'}{}_{(A|A'|}{}^{B'}\nabla_{B|B'|}\kappa_{C)}\nonumber\\
& -  \kappa^{D}\nabla_{(A}{}^{A'}\Qop_{B|A'|C)D}
 + \tfrac{1}{8} \nabla_{(A}{}^{A'}G\nabla_{B|A'|}\kappa_{C)}\nonumber\\
& + \tfrac{1}{2} \nabla_{(A}{}^{A'}G_{B}{}^{D}{}_{|A'}{}^{B'}\nabla_{DB'|}\kappa_{C)}\\
={}&\Psi_{ABCD} \vartheta \kappa^{D}
 -  \tfrac{1}{4} G \Psi_{ABCD} \kappa^{D}
 -  \kappa^{D} \nabla_{(A}{}^{A'}\Qop_{B|A'|C)D}\nonumber\\
& + \tfrac{1}{2} \kappa^{D} G_{(AB}{}^{A'B'}\Phi_{C)DA'B'}.
\end{align*}
The above expression holds for all $\kappa^A$, and therefore we can
conclude that 
\[
\vartheta(\Psi_{ABCD})=- \tfrac{1}{4} G \Psi_{ABCD}
 -  \nabla_{(A}{}^{A'}\Qop_{B|A'|C)D}
 + \tfrac{1}{2} G_{(AB}{}^{A'B'}\Phi_{C)DA'B'}.
\]
The symmetry of $\Psi_{ABCD}$ can be used to simplify this last expression
---the trace of the right hand side can be shown to vanish due to the
commutators.

\usemedskip
If we compute the variation of
\[
\Phi_{BAA'B'} \kappa^{A}=- \nabla^{A}{}_{(A'}\nabla_{|A|B')}\kappa_{B}
\]
we get
\begin{align}
\Phi_{BAA'B'} \vartheta \kappa^{A} + \vartheta(\Phi_{BAA'B'}) \kappa^{A}={}&- \nabla^{A}{}_{(A'}\nabla_{|A|B')}\vartheta \kappa_{B}
 + \tfrac{1}{4} G \nabla^{A}{}_{(A'}\nabla_{|A|B')}\kappa_{B}\nonumber\\
& -  \tfrac{1}{2} G^{AC}{}_{(A'}{}^{C'}\nabla_{|A|B')}\nabla_{CC'}\kappa_{B}
 + \tfrac{1}{2} G^{AC}{}_{(A'}{}^{C'}\nabla_{|AC'}\nabla_{C|B')}\kappa_{B}\nonumber\\
& + \Qop^{A}{}_{(A'|A}{}^{C}\nabla_{C|B')}\kappa_{B}
 + \bar\Qop_{(A'}{}^{A}{}_{B')}{}^{C'}\nabla_{AC'}\kappa_{B}\nonumber\\
& -  \kappa^{A}\nabla^{C}{}_{(A'}\Qop_{|C|B')BA}
 + \tfrac{1}{8} \nabla^{A}{}_{(A'}G\nabla_{|A|B')}\kappa_{B}\nonumber\\
& + \tfrac{1}{2}
\nabla^{A}{}_{(A'}G_{|A|}{}^{C}{}_{B')}{}^{C'}\nabla_{CC'}\kappa_{B} \nonumber\\
={}&\Phi_{BAA'B'} \vartheta \kappa^{A}
 + G_{BAA'B'} \Lambda \kappa^{A}
 -  \tfrac{1}{4} G \Phi_{BAA'B'} \kappa^{A}\nonumber\\
& + \tfrac{1}{2} G^{CD}{}_{A'B'} \Psi_{BACD} \kappa^{A}
 -  \kappa^{A} \nabla^{C}{}_{(A'}\Qop_{|C|B')BA}. \nonumber
\end{align}
The last relation holds for all $\kappa^A$, and therefore we can obtain
an expression for $\vartheta \Phi_{ABA'B'}$.

\usemedskip
Now, using the definition of $\Qop_{ABCA'}$, commuting derivatives and
exploiting the irreducible decomposition of the various fields
involved one gets
\begin{align}
\nabla_{AA'}\Qop^{CA'}{}_{BC}={}&- \tfrac{1}{2} G_{B}{}^{CA'B'} \Phi_{ACA'B'}
 -  \tfrac{1}{2} G_{A}{}^{CA'B'} \Phi_{BCA'B'} \nonumber \\
& + \nabla_{CA'}\Qop_{A}{}^{A'}{}_{B}{}^{C}
 + \epsilon_{AB} \nabla^{CA'}\Qop^{D}{}_{A'CD}. \label{eq:nablaQopEq2}
\end{align} 
If we compute the variation of
\[
\Lambda \kappa_{A}=\tfrac{1}{3} \nabla_{(A}{}^{A'}\nabla_{B)A'}\kappa^{B}
\]
we get, after a lengthy computation, that 
\begin{subequations}
\begin{align}
\Lambda \vartheta \kappa_{A} + \vartheta(\Lambda) \kappa_{A}={}&\tfrac{1}{6} \kappa^{B} \nabla_{AA'}\Qop^{CA'}{}_{BC}
 -  \tfrac{1}{6} \nabla_{AA'}\nabla_{B}{}^{A'}\vartheta \kappa^{B}
 + \tfrac{1}{24} G \nabla_{AA'}\nabla_{B}{}^{A'}\kappa^{B}\nonumber\\
& + \tfrac{1}{6} \bar\Qop^{B'}{}_{BA'B'} \nabla_{A}{}^{A'}\kappa^{B}
 -  \tfrac{1}{12} G_{BCA'B'} \nabla_{A}{}^{B'}\nabla^{CA'}\kappa^{B}
 -  \tfrac{1}{48} \nabla_{A}{}^{A'}G \nabla_{BA'}\kappa^{B}\nonumber\\
& -  \tfrac{1}{6} \nabla_{BA'}\nabla_{A}{}^{A'}\vartheta \kappa^{B}
 + \tfrac{1}{24} G \nabla_{BA'}\nabla_{A}{}^{A'}\kappa^{B}
 + \tfrac{1}{6} \bar\Qop^{B'}{}_{AA'B'} \nabla_{B}{}^{A'}\kappa^{B}\nonumber\\
& -  \tfrac{1}{12} G_{ACA'B'} \nabla_{B}{}^{B'}\nabla^{CA'}\kappa^{B}
 + \tfrac{1}{48} \nabla_{AA'}\kappa_{B} \nabla^{BA'}G
 -  \tfrac{1}{6} \kappa^{B} \nabla_{CA'}\Qop_{A}{}^{A'}{}_{B}{}^{C}\nonumber\\
& -  \tfrac{1}{6} \Qop_{AA'BC} \nabla^{CA'}\kappa^{B}
 -  \tfrac{1}{6} \Qop_{BA'AC} \nabla^{CA'}\kappa^{B}
 + \tfrac{1}{12} \nabla_{AB'}G_{BCA'}{}^{B'} \nabla^{CA'}\kappa^{B}\nonumber\\
& + \tfrac{1}{12} \nabla_{BB'}G_{ACA'}{}^{B'} \nabla^{CA'}\kappa^{B}
 + \tfrac{1}{12} G_{BCA'B'} \nabla^{CB'}\nabla_{A}{}^{A'}\kappa^{B}\nonumber\\
& + \tfrac{1}{12} G_{ACA'B'} \nabla^{CB'}\nabla_{B}{}^{A'}\kappa^{B}\nonumber\\
={}&\Lambda \vartheta \kappa_{A}
 -  \tfrac{1}{4} G \Lambda \kappa_{A}
 + \tfrac{1}{6} G_{A}{}^{CA'B'} \Phi_{BCA'B'} \kappa^{B}
 + \tfrac{1}{6} \kappa^{B} \nabla_{AA'}\Qop^{CA'}{}_{BC}\nonumber\\
& -  \tfrac{1}{6} \kappa^{B}
\nabla_{CA'}\Qop_{A}{}^{A'}{}_{B}{}^{C}\nonumber \\
={}&\Lambda \vartheta \kappa_{A}
 -  \tfrac{1}{4} G \Lambda \kappa_{A}
 + \tfrac{1}{12} G^{BCA'B'} \Phi_{BCA'B'} \kappa_{A}
 -  \tfrac{1}{6} \kappa_{A}
 \nabla_{CA'}\Qop^{BA'}{}_{B}{}^{C}.\nonumber 
\end{align}
\end{subequations}

In the last equality we have used the relation \eqref{eq:nablaQopEq2}
and the irreducible decomposition of $ G_{A}{}^{CA'B'}
\Phi_{BCA'B'}$. From here we can deduce an expression for
$\vartheta\Lambda$.

\usemedskip
We summarise the discussion of this section in the following:

\begin{theorem}
The modified variation of the curvature spinors is given by 
\begin{align*}
\vartheta\Psi_{ABCD}={}&- \tfrac{1}{4} G \Psi_{ABCD}
 -  \nabla_{(A}{}^{A'}\Qop_{B|A'|CD)}
 + \tfrac{1}{2} G_{(AB}{}^{A'B'}\Phi_{CD)A'B'},\\
\vartheta\Phi_{ABA'B'}={}&G_{ABA'B'} \Lambda
 -  \tfrac{1}{4} G \Phi_{ABA'B'}
 + \tfrac{1}{2} G^{CD}{}_{A'B'} \Psi_{ABCD}
 -  \nabla^{C}{}_{(A'}\Qop_{|C|B')AB},\\
\vartheta\Lambda={}&- \tfrac{1}{4} G \Lambda
 + \tfrac{1}{12} G^{BCA'B'} \Phi_{BCA'B'}
 -  \tfrac{1}{6} \nabla_{CA'}\Qop^{BA'}{}_{B}{}^{C}.
\end{align*}
\end{theorem}

\begin{remark}
For a pure gauge transformation \eqref{eq:puregauge}, we get after a
lengthy but straightforward calculation using commutators, that 
\begin{align*}
\vartheta(\Lambda)={}&(\mathcal{L}_{\xi}\Lambda),\\
\vartheta(\Phi_{AB}{}^{A'B'})={}&(\mathcal{L}_{\xi}\Phi)_{AB}{}^{A'B'}
 -  \Phi^{C}{}_{(A}{}^{C'(A'}\nabla_{B)}{}^{B')}\xi_{CC'}
 -  \Phi^{C}{}_{(A}{}^{C'(A'}\nabla_{|CC'|}\xi_{B)}{}^{B')},\\
\vartheta(\Psi_{ABCD})={}&(\mathcal{L}_{\xi}\Psi)_{ABCD}
- \Psi_{ABCD} \nabla^{FA'}\xi_{FA'},
\end{align*}
where \footnote{The primed indices are moved up after the Lie derivative is taken to allow the symmetrizations to be written nicely.}
\begin{align*}
(\mathcal{L}_{\xi}\Phi)_{AB}{}^{A'B'}\equiv{}&\xi^{CC'} \nabla_{CC'}\Phi_{AB}{}^{A'B'}
 + 2 \Phi^{C}{}_{(A}{}^{C'(A'}\nabla_{B)}{}^{B')}\xi_{CC'},\\
(\mathcal{L}_{\xi}\Psi)_{ABCD}\equiv{}&\xi^{FA'} \nabla_{FA'}\Psi_{ABCD}
 + 2 \Psi_{(ABC}{}^{F}\nabla_{D)}{}^{A'}\xi_{FA'}.
\end{align*}
In this last calculation we have used the Bianchi identity in the form
\begin{align*}
\nabla^{D}{}_{A'}\Psi_{ABCD}={}&\nabla_{(A}{}^{B'}\Phi_{B)CA'B'}
 +  \epsilon_{C(A}\nabla_{B)A'}\Lambda.
\end{align*}
\end{remark}

\section{Variations of space-spinor expressions}\label{sec:spacespinors}

The analysis of Sections \ref{Section:SLCalculus},
\ref{Section:CovDevs} and \ref{Section:Curvature} can be adapted to
consider variations of spinorial fields in a space-spinor
formalism. This formalism can be used to analyse variational problems
in 3-dimensional Riemannian manifolds.

\subsection{Basic formalism}
In what follows, let $(\mathcal{S},h_{ij})$ denote a 3-dimensional
Riemannian manifold with negative-definite metric. On
$(\mathcal{S},h_{ij})$ we assume the existence of a spinor structure
with an antisymmetric spinor $\epsilon_{AB} $. In addition, we assume
that the spinor structure is endowed with an Hermitian product. It
follows from this assumption that there exists an Hermitian spinor
$\varpi_{AA'}$ such given two spinors $\xi_A$ and $\eta_B$ the
Hermitian inner product can be expressed as
\[
\xi_A \hat{\eta}^A \equiv \varpi_{AA'} \bar{\eta}^{A'} \xi^A.
\]
The spinor $\hat{\eta}^A$ defined by the above relation is called the
Hermitian conjugate of $\eta^A$.

Let $e_{\bf k}{}^l$, $\omega^{\bf k}{}_l$ denote, respectively, an
orthonormal frame and coframe of $(\mathcal{S},h_{ij})$ and let
$\epsilon^{\bf A}{}_B$ denote a normalised spin dyad such that the
components of $\epsilon_{AB}$ and $\varpi_{AA'}$ are given,
respectively, by
\[
\epsilon_{\bf AB} = \left(  
\begin{array}{cc}
0 & 1 \\
-1 & 0
\end{array}
\right),\qquad
\varpi_{\bf AA'} = \left(  
\begin{array}{cc}
1 & 0 \\
0 & 1
\end{array}
\right).
\]
The transformations of the spin dyad respecting the above expressions
is given by $SU(2,\mathbb{C})$ matrices $O_{\bf A}{}^{\bf B}$.

The correspondence between spatial tensors and spinors is realised by the \emph{spatial Infeld-van der Waerden
symbols} $\sigma_{\bf k}{}^{\bf AB}$ and $\sigma^{\bf k}{}_{\bf AB}$. Given an arbitrary $ v^k \in T \mathcal{S}$
and $\beta_k \in T^* \mathcal{S}$ one has that
\[
v^{\bf k} \mapsto v^{\bf AB} = v^{\bf k} \sigma_{\bf k}{}^{\bf AB},
\qquad \beta_{\bf k} \mapsto \beta_{\bf AB} =\beta_{\bf k}\sigma^{\bf k}{}_{\bf AB}, 
\]
where  
\[
v^{\bf k} \equiv v^k \omega^{\bf k}{}_k, \qquad \beta_{\bf k} \equiv
\beta_k e_{\bf k}{}^k.
\]
 In more explicit terms, the
correspondence is 
\[
(v^{\bf 1}, v^{\bf 2}, v^{\bf 3}) \mapsto
\frac{1}{\sqrt{2}}
\left(\begin{array}{cc}
- v^{\bf 1} - \mbox{i} v^{\bf 2} & v^{\bf 3}\\
v^{\bf 3} & v^{\bf 1} - \mbox{i} v^{\bf 2}
\end{array}\right)
, \qquad
(\beta_{\bf 1}, \beta_{\bf 2}, \beta_{\bf 3})\mapsto
\frac{1}{\sqrt{2}}
\left(\begin{array}{cc}
- \beta_{\bf 1} + \mbox{i} \beta_{\bf 2} & \beta_{\bf 3}\\
\beta_{\bf 3} & \beta_{\bf 1} + \mbox{i} \beta_{\bf 2}
\end{array}\right).
\]

From these, we define the spatial soldering form to be
\begin{subequations}
\begin{align}
\sigma_{k}{}^{AB}\equiv{}&\omega^{\bf l}{}_{k} \epsilon_{{\bf C}}{}^{A} \epsilon_{{\bf D}}{}^{B} 
\sigma_{{\bf l}}{}^{{\bf C} {\bf D}},\label{eq:SSsigmadef}\\
\sigma^{k}{}_{AB}\equiv{}&\epsilon^{\bf C}{}_{A} \epsilon^{\bf D}{}_{B} e_{{\bf l}}{}^{k} \sigma^{{\bf l}}{}_{{\bf C} {\bf D}}.
\end{align}
\end{subequations}
As we allow the spinor and tensor frames to be independent, the
soldering form will therefore be frame dependent.  However, we will
always have the universal relations
\begin{subequations}
\begin{align}
\sigma_{k}{}^{CD} \sigma^{l}{}_{CD}={}&\delta_{k}{}^{l},\label{eq:SSsigmatodelta}\\
h_{kl}={}&\sigma_{k}{}^{AB} \sigma_{l}{}^{CD} \epsilon_{CA} \epsilon_{DB}.
\end{align}
\end{subequations}

The Hermitian conjugate of 
\[
\phi_{A}= \phi_{\bf 0}\epsilon^{\bf 0}{}_{A}
 +  \phi_{\bf 1}\epsilon^{\bf 1}{}_{A}
\]
is given by
\[
\hat{\phi}_{A}= -  \bar{\phi}_{\bf 1'}\epsilon^{\bf 0}{}_{A}
+  \bar{\phi}_{\bf 0'}\epsilon^{\bf 1}{}_{A}.
\]
It clearly follows that 
\[
\hat{\hat{\phi}}_{A}=-\phi_{A}.
\]
The Hermitian conjugation can be extended to higher valence space
spinors by requiring that the conjugate of a product equals the
product of conjugates. We also get
\[
\hat{\hat{\mu}}_{A_1\cdots A_k} = (-1)^k \hat{\mu}_{A_1\cdots A_k}.
\]
Furthermore, it is important to note
\begin{equation}
\hat{\epsilon}_{AB}=\epsilon_{AB},\qquad 
\hat{\sigma}^{a}{}_{AB}=- \sigma^{a}{}_{AB}.\label{eq:epsilonsigmaHermitian}
\end{equation}

\subsection{Basic variational formulae}
As in the case of standard spacetime spinors, we can compute the
variations of the frames and the inverse metrics from the relations
\begin{align*}
\delta(e_{{\bf i}}{}^{l}) ={}& - e_{{\bf i}}{}^{j} e_{{\bf k}}{}^{l} \delta \omega^{{\bf k}}{}_{j},\\
\delta(\epsilon_{{\bf A}}{}^{D}) ={}& - \epsilon_{{\bf A}}{}^{B} \epsilon_{{\bf C}}{}^{D} \delta \epsilon^{{\bf C}}{}_{B},\\
\delta(h^{kl}) ={}& - (\delta h_{ij}) h^{ik} h^{jl},\\
\delta(\epsilon^{CD}) ={}& - \delta \epsilon_{AB} \epsilon^{AC} \epsilon^{BD}.
\end{align*}
Likewise, from the relation \eqref{eq:SSsigmatodelta} we get
\[
\delta(\sigma^{l}{}_{AB}) = - \sigma^{k}{}_{AB} \sigma^{l}{}_{CD} \delta \sigma_{k}{}^{CD}.
\]

We can also split the variation of the coframes in terms of the
variation of the metric and spin metric and gauge pieces
\begin{align*}
\delta \omega^{{\bf m}}{}_{a} ={}& - e_{{\bf h}}{}^{b}  h^{{\bf h} {\bf m}}T_{ab} + \tfrac{1}{2} e_{{\bf h}}{}^{b} h^{{\bf h} {\bf m}} \delta h_{ab},\\
\delta \epsilon^{{\bf P}}{}_{A} ={}& - \epsilon_{{\bf H}}{}^{B}  \epsilon^{{\bf H} {\bf P}}S_{AB} -  \tfrac{1}{2} \epsilon_{{\bf H}}{}^{B} \epsilon^{{\bf H} {\bf P}} \delta \epsilon_{AB} ,
\end{align*}
where the tensor and spinor frame gauge fields are
\begin{align*}
T_{ab} \equiv{}& h_{{\bf c} {\bf d}}\omega^{\bf d}{}_{[a}  \delta \omega^{{\bf c}}{}_{b]},&
S_{AB} \equiv{}& \omega^{\bf D}{}_{(A} \delta \omega^{{\bf C}}{}_{B)} \epsilon_{{\bf C} {\bf D}}.
\end{align*}
A calculation following the same principles as for the spacetime
version starting from the relation \eqref{eq:SSsigmadef} gives the variation of the
spatial soldering form:
\begin{align*}
\delta \sigma_{k}{}^{AB}={}&- T_{k}{}^{l} \sigma_{l}{}^{AB}
 + \tfrac{1}{2} \sigma^{lAB} \delta h_{kl}
 - 2 \sigma_{k}{}^{(A|C|}S^{B)}{}_{C}
 + \sigma_{k}{}^{(A|C|}\delta \epsilon^{B)}{}_{C}.
\end{align*}
The irreducible parts are given by
\begin{align*}
\sigma^{k(CD}\delta \sigma_{k}{}^{AB)}={}&\tfrac{1}{2} \delta h^{(ABCD)},\\
\sigma^{k(C}{}_{B}\delta \sigma_{k}{}^{A)B}={}&T^{AC}
 - 2 S^{AC},\\
\sigma^{k}{}_{CD} \delta \sigma_{k}{}^{CD}={}&\tfrac{1}{2} \delta h^{CD}{}_{CD}
 + \tfrac{3}{2} \delta \epsilon^{C}{}_{C},
\end{align*}
where
\begin{align*}
T_{AB}\equiv{}&T_{kl} \sigma^{k}{}_{A}{}^{C} \sigma^{l}{}_{BC},\\
\delta h_{ABCD}\equiv{}&\sigma^{k}{}_{AB} \sigma^{l}{}_{CD} \delta h_{kl}.
\end{align*}

We can now use this to see how the variation of vectors and covectors
in space-spinor and tensor form differ:
\begin{align*}
\sigma_{k}{}^{AB} \delta\zeta^k={}&\delta(\zeta^{AB})
 -  \tfrac{1}{2} \delta \epsilon^{C}{}_{C} \zeta^{AB}
 -  \tfrac{1}{2} \delta h^{AB}{}_{CD} \zeta^{CD}
 + T^{(A|C|}\zeta^{B)}{}_{C}
 - 2 S^{(A|C|}\zeta^{B)}{}_{C},\\
\sigma^{k}{}_{AB} \delta\xi_k={}&\delta(\xi_{AB})
 + \tfrac{1}{2} \delta \epsilon^{C}{}_{C} \xi_{AB}
 + \tfrac{1}{2} \delta h_{ABCD} \xi^{CD}
 + T_{(A}{}^{C}\xi_{B)C}
 - 2 S_{(A}{}^{C}\xi_{B)C},
\end{align*}
where $\zeta^k=\sigma^{k}{}_{CD}\zeta^{CD}$ and
$\xi_k=\sigma_{k}{}^{CD}\zeta_{CD}$.  This leads us to define a
modified variation that cancels the gauge terms and the variation of the
spin metric.

\begin{definition}
For valence 1 space spinors we define the modified variation operator $\vartheta$ via
\begin{align*}
\vartheta(\phi_{A})\equiv{}&\delta(\phi_{A})
 + \tfrac{1}{4} \delta \epsilon^{B}{}_{B} \phi_{A}
 +  \tfrac{1}{2} T_{A}{}^{B} \phi_{B}
 - S_{A}{}^{B} \phi_{B},\\
\vartheta(\phi^{A})\equiv{}&\delta(\phi^{A})
 -  \tfrac{1}{4} \delta \epsilon^{B}{}_{B} \phi^{A}
 -  \tfrac{1}{2} T^{A}{}_{B} \phi^{B}
 + S^{A}{}_{B} \phi^{B}.
\end{align*}
These relations extend to higher valence spinors via the Leibnitz rule.
\end{definition}

In the same way as for the spacetime variations, we get a relation
between $\vartheta$ and spin frame component variation:
\begin{align}
\vartheta \phi_{A}={}&\epsilon^{\bf B}{}_{A} \delta(\phi_{{\bf B}})
 + \tfrac{1}{2} T_{A}{}^{B} \phi_{B}.\label{eq:SSvarthetatocomps}
\end{align}

The reality of $T_{ab}$  and \eqref{eq:epsilonsigmaHermitian} gives
\[
\widehat{T}_{AB}= T_{AB}.
\]
Expanding the frame index in equation \eqref{eq:SSvarthetatocomps} and taking Hermitian conjugate yields
\begin{align*}
\widehat{\vartheta \phi}_{A}={}&\epsilon^{\bf 1}{}_{A} \delta(\bar{\phi}_{{\bf 0'}})
 -  \epsilon^{\bf 0}{}_{A} \delta(\bar{\phi}_{{\bf 1'}})
 + \tfrac{1}{2} T_{A}{}^{B} \hat{\phi}_{B}\nonumber\\
={}&\epsilon^{\bf 0}{}_{A} \delta(\hat{\phi}_{{\bf 0}})
 + \epsilon^{\bf 1}{}_{A} \delta(\hat{\phi}_{{\bf 1}})
 + \tfrac{1}{2} T_{A}{}^{B} \hat{\phi}_{B}\nonumber\\
={}&\vartheta(\hat\phi)_{A}.
\end{align*}
Hence, the operation of Hermitian conjugation and the modified variation $\vartheta$ commute.

\subsection{Variations of the spatial connection}
Let $\mathcal{R}_{ABCD}$ denote the space spinor version of the trace
free Ricci tensor, and let $\mathcal{R}$ be the Ricci scalar.  Define
\begin{align*}
H^{ABCD}\equiv{}&\delta h^{(ABCD)},\\
H\equiv{}&\delta h_{AB}{}^{AB},\\
\Qop_{ABCD}\equiv {}&- \tfrac{1}{2} D_{(C}{}^{F}\delta h_{D)FAB},\\
F^{AB}\equiv{}&- \tfrac{1}{2} D^{AB}\delta h^{CD}{}_{CD}
 + D_{CD}\delta h^{ABCD}.
\end{align*}
Similarly to the case of spacetime spinors, we can compute the variation of a covariant derivative.
\begin{theorem}
The variation of a covariant space-spinor derivative is given by
\begin{align*}
\vartheta(D_{AB}\kappa_{C})={}&
 D_{AB}\vartheta \kappa_{C}
 +\Qop_{ABCD} \kappa^{D}
 -  \tfrac{1}{2} \delta h_{ABDF} D^{DF}\kappa_{C}.
\end{align*}
\end{theorem}

We also get
\begin{align*}
\Qop_{A}{}^{C}{}_{BC}={}&- \tfrac{1}{6} D_{AB}H
 + \tfrac{1}{4} D_{CD}H_{AB}{}^{CD},\\
F_{AB}={}&- \tfrac{1}{6} D_{AB}H
 + D_{CD}H_{AB}{}^{CD},\\
 D_{DF}H_{ABC}{}^{F}={}&2 \Qop_{(ABC)D}
 + 2 \epsilon_{D(A}\Qop_{B}{}^{F}{}_{C)F}
 +  \tfrac{1}{2} \epsilon_{D(A}D_{BC)}H.
\end{align*}

\subsection{Diffeomorphism dependence}
To analyse the dependence of the formalism on diffeomorphisms, we
proceed in the same way as in
Section~\ref{sec:diffeomorphisms}. Accordingly, let $\phi_\lambda$ be
a one parameter group of diffeomorphisms generated by a vector field
$\xi^a$. Now, let $h_{ab}[\lambda]=\phi^*_{-\lambda}\mathring
h_{ab}$. All members of the family $h_{ab}[\lambda]$ will have 
the same geometric content and we get
\begin{align}
\delta h_{ab}={}&\mathcal{L}_\xi h_{ab}
=2 D_{(a}\xi_{b)}.\label{eq:spatialpuregauge}
\end{align}
Moreover, one has that 
\begin{align*}
\Qop_{ABCD}={}&- \tfrac{1}{2} D_{(C}{}^{F}D_{D)F}\xi_{AB}
 -  \tfrac{1}{2} D_{(C}{}^{F}D_{|AB|}\xi_{D)F},\\
F^{AB}={}&D_{CD}D^{CD}\xi^{AB}
- \tfrac{1}{3} \mathcal{R} \xi^{AB}
 -  \mathcal{R}^{AB}{}_{CD} \xi^{CD}.
\end{align*}
Again, we see that $F^{AB}$ can be interpreted as a gauge source
function for the linearised diffeomorphisms, but this time one needs
to solve an elliptic equation instead of a wave equation to obtain
$\xi^{AB}$ from $F^{AB}$.

\subsection{Variations of the spatial curvature}
By computing the variation of the commutator relations
\begin{subequations}
\begin{align}
\mathcal{R} \kappa_{A}={}&8 D_{(A}{}^{C}D_{B)C}\kappa^{B},\label{eq:SScommutatorR0}\\
\mathcal{R}_{ABCD} \kappa^{D}={}&2 D_{(A}{}^{D}D_{B|D|}\kappa_{C)},\label{eq:SScommutatorR4}
\end{align}
\end{subequations}
we get, after calculations similar to those carried out in the spacetime case, the variation of the curvature.
\begin{theorem}
The variation of the spatial curvature spinors is given by
\begin{align*}
\vartheta(\mathcal{R})={}&- \tfrac{1}{3} H \mathcal{R}
 -  H^{BCDF} \mathcal{R}_{BCDF}
 - 4 D_{CD}\Qop^{BC}{}_{B}{}^{D},\\
\vartheta(\mathcal{R}_{ABCD})={}&- \tfrac{1}{12} H_{ABCD} \mathcal{R}
 -  \tfrac{1}{3} H \mathcal{R}_{ABCD}
 + 2 D_{(A}{}^{F}\Qop_{B|F|CD)}
 + \tfrac{1}{2} H_{(AB}{}^{FH}\mathcal{R}_{CD)FH}.
\end{align*}
\end{theorem}
\begin{proof}
Computing the variation of relation \eqref{eq:SScommutatorR0} gives
\begin{align*}
\mathcal{R} \vartheta \kappa_{A} + \vartheta(\mathcal{R}) \kappa_{A}={}&-4 D_{AC}D_{B}{}^{C}\vartheta \kappa^{B}
 + \tfrac{4}{3} H D_{AC}D_{B}{}^{C}\kappa^{B}
 + 4 \Qop_{B}{}^{D}{}_{CD} D_{A}{}^{C}\kappa^{B}
 + 4 \kappa^{B} D_{AD}\Qop^{CD}{}_{BC}\nonumber\\
& - 2 H_{BCDF} D_{A}{}^{F}D^{CD}\kappa^{B}
 -  \tfrac{2}{3} D_{A}{}^{B}H D_{BC}\kappa^{C}
 - 4 D_{BC}D_{A}{}^{C}\vartheta \kappa^{B}\nonumber\\
& + \tfrac{4}{3} H D_{BC}D_{A}{}^{C}\kappa^{B}
 + 4 \Qop_{A}{}^{D}{}_{CD} D_{B}{}^{C}\kappa^{B}
 - 2 H_{ACDF} D_{B}{}^{F}D^{CD}\kappa^{B}\nonumber\\
& + \tfrac{2}{3} D_{AC}\kappa_{B} D^{BC}H
 - 4 \kappa^{B} D_{CD}\Qop_{A}{}^{C}{}_{B}{}^{D}
 - 4 \Qop_{ACBD} D^{CD}\kappa^{B}
 - 4 \Qop_{BCAD} D^{CD}\kappa^{B}\nonumber\\
& + 2 D_{AF}H_{BCD}{}^{F} D^{CD}\kappa^{B}
 + 2 D_{BF}H_{ACD}{}^{F} D^{CD}\kappa^{B}
 + 2 H_{BCDF} D^{DF}D_{A}{}^{C}\kappa^{B}\nonumber\\
& + 2 H_{ACDF} D^{DF}D_{B}{}^{C}\kappa^{B}\\
={}&\mathcal{R} \vartheta \kappa_{A}
 -  \tfrac{1}{3} H \mathcal{R} \kappa_{A}
 - 2 H_{A}{}^{CDF} \mathcal{R}_{BCDF} \kappa^{B}
 + 4 \kappa^{B} D_{AD}\Qop^{CD}{}_{BC}\nonumber\\
& - 4 \kappa^{B} D_{CD}\Qop_{A}{}^{C}{}_{B}{}^{D}\\
={}&\mathcal{R} \vartheta \kappa_{A}
 -  \tfrac{1}{3} H \mathcal{R} \kappa_{A}
 -  H^{BCDF} \mathcal{R}_{BCDF} \kappa_{A}
 - 4 \kappa_{A} D_{CD}\Qop^{BC}{}_{B}{}^{D}.
\end{align*}
Computing the variation of relation \eqref{eq:SScommutatorR4} gives
\begin{align*}
\mathcal{R}_{ABCD} \vartheta \kappa^{D} + \vartheta(\mathcal{R}_{ABCD}) \kappa^{D}={}&2 D_{(A}{}^{D}D_{B|D|}\vartheta \kappa_{C)}
 -  \tfrac{2}{3} H D_{(A}{}^{D}D_{B|D|}\kappa_{C)}\nonumber\\
& + H_{(A}{}^{DFH}D_{B|D}D_{FH|}\kappa_{C)}
 -  H_{(A}{}^{DFH}D_{|DF|}D_{B|H|}\kappa_{C)}\nonumber\\
& - 2 \Qop_{(A}{}^{D}{}_{B}{}^{F}D_{|DF|}\kappa_{C)}
 - 2 \Qop_{(A}{}^{D}{}_{|D|}{}^{F}D_{B|F|}\kappa_{C)}\nonumber\\
& + 2 \kappa^{D}D_{(A}{}^{F}\Qop_{B|F|C)D}
 -  \tfrac{1}{3} D_{(A}{}^{D}HD_{B|D|}\kappa_{C)}\nonumber\\
& -  D_{(A}{}^{D}H_{B|D}{}^{FH}D_{FH|}\kappa_{C)}\\
={}&\mathcal{R}_{ABCD} \vartheta \kappa^{D}
 -  \tfrac{1}{12} H_{ABCD} \mathcal{R} \kappa^{D}
 -  \tfrac{1}{3} H \mathcal{R}_{ABCD} \kappa^{D}\nonumber\\
& + 2 \kappa^{D} D_{(A}{}^{F}\Qop_{B|F|C)D}
 + \tfrac{1}{2} \kappa^{D} H_{(AB}{}^{FH}\mathcal{R}_{C)DFH}.
\end{align*}
\end{proof}

\begin{remark}
For a pure gauge transformation \eqref{eq:spatialpuregauge}, we get
\begin{align*}
\vartheta(\mathcal{R})={}&\mathcal{L}_{\xi}\mathcal{R},\\
\vartheta(\mathcal{R}_{ABCD})={}&\mathcal{L}_{\xi}\mathcal{R}_{ABCD}
 -  \mathcal{R}_{(AB}{}^{FH}D_{CD)}\xi_{FH}
 -  \mathcal{R}_{(AB}{}^{FH}D_{|FH|}\xi_{CD)},
\end{align*}
where
\begin{align*}
\mathcal{L}_{\xi}\mathcal{R}_{ABCD}={}&\xi^{FH} D_{FH}\mathcal{R}_{ABCD}
 + 2 \mathcal{R}_{(AB}{}^{FH}D_{CD)}\xi_{FH}.
\end{align*}
\end{remark}

\section*{Acknowledgements}
We thank L. B. Szabados for helpful conversations at the beginning of this project. We also thank L. Andersson and S. Aksteiner for discussions regarding gauge dependence.
TB was supported by the Engineering and Physical Sciences Research Council [grant number EP/J011142/1].

\appendix

\section{Rotations}
The purpose of this appendix is to discuss some issues related to the
gauge in the frame and spin dyad formalisms.

\subsection{Lorentz transformations}
\label{Section:LorentzTransformations}
As it is well known, the metric $g_{ab}$ is not determined in a unique
way by the orthonormal coframe $\omega^{\bf a}{}_a$. Any
other coframe related to $\omega^{\bf a}{}_a$ by means of a
Lorentz transformation ---i.e. a matrix $( \Lambda^{\bf a}{}_{\bf b})$
such that 
\begin{equation}
\eta_{\bf a b} \Lambda^{\bf a}{}_{\bf c} \Lambda^{\bf b}{}_{\bf d}
=\eta_{\bf c d}.
\label{LorentzTransformation}
\end{equation}
It follows that $
\grave{\omega}{}^{\bf a}{}_a \equiv \Lambda^{\bf a}{}_{\bf b}
\omega^{\bf b}{}_a$ is also orthonormal with respect to $g_{ab}$ and one can write
$g_{ab} = \eta_{\bf ab} \grave{\omega}^{\bf a}{}_a \grave{\omega}^{\bf
b}{}_b$. The associated orthonormal frame is $\grave{e}_{\bf
a}{}^a = \Lambda_{\bf a}{}^{\bf b} e_{\bf b}{}^a $ with
$(\Lambda_{\bf a}{}^{\bf b}) \equiv ( \Lambda^{\bf a}{}_{\bf b})^{-1}$
where the last expression is a relation between matrices.

\usemedskip
The discussion in the previous paragraph can be extended to include
spinors. Making use of the Infeld-van der Waerden symbols, equation
\eqref{LorentzTransformation} can be rewritten as
\[
\epsilon_{\bf AB} \epsilon_{\bf A'B'} \Lambda^{\bf AA'}{}_{\bf CC'}
\Lambda^{\bf BB'}{}_{\bf DD'} =\epsilon_{\bf CD} \epsilon_{\bf C'D'},
\]
with $\Lambda^{\bf AA'}{}_{\bf CC'} \equiv   \sigma_{\bf a}{}^{\bf
  AA'} \sigma^{\bf c}{}_{\bf CC'}       \Lambda^{\bf
  a}{}_{\bf c}$. It can be shown that the spinorial components
$\Lambda^{\bf AA'}{}_{\bf CC'}$ can be decomposed as
\[
\Lambda^{\bf AA'}{}_{\bf CC'} = \Lambda^{\bf A}{}_{\bf C}
\bar{\Lambda}^{\bf A'}{}_{\bf C'},
\]
where $(\Lambda^{\bf A}{}_{\bf C} )$ is a $SL(2,\mathbb{C})$
matrix. The latter naturally induces a change of spinorial basis via
the relations 
\[
\grave{\epsilon}_{\bf A}{}^A = \Lambda_{\bf A}{}^{\bf B} \epsilon_{\bf
  B}{}^A, \qquad \grave{\epsilon}^{\bf A}{}_A = \Lambda^{\bf A}{}_{\bf
  B} \epsilon^{\bf B}{}_A,
\]
with $(\Lambda_{\bf A}{}^{\bf B}) \equiv (\Lambda^{\bf A}{}_{\bf
  B})^{-1}$. Crucially, one has that
\[
\epsilon_{\bf AB} = \Lambda^{\bf C}{}_{\bf A} \Lambda^{\bf D}{}_{\bf
  B} \epsilon_{\bf CD}, \qquad \epsilon^{\bf AB} = \Lambda_{\bf
  C}{}^{\bf A} \Lambda_{\bf D}{}^{\bf B} \epsilon^{\bf CD}. 
\]

\subsection{$O(3)$-rotations}
Given a 3-dimensional negative-definite Riemannian metric $h_{ij}$ and
an associated orthonormal coframe $\omega^{\bf i}{}_k$ one has that
\[
h_{ij} =-\delta_{\bf i j} \omega^{\bf i}{}_i \omega^{\bf j}{}_j.
\]
Any other coframe $\grave{\omega}^{\bf i}{}_k$ related to the coframe
$\omega^{\bf i}{}_k$ through the relation $\grave{\omega}^{\bf
  i}{}_k = O_{\bf j}{}^{\bf i} \omega^{\bf j}{}_k$, where $(O_{\bf j}{}^{\bf
  i})$ is a $O(3)$-matrix, gives rise to the same metric. The
defining condition for $(O^{\bf
  i}{}_{\bf j})$ can be expressed as
\[
\delta_{\bf i j} = \delta_{\bf k l} O_{\bf i}{}^{\bf k}{} O_{\bf j}{}^{\bf
  l}.
\]

A direct calculation using the definition of the Hermitian product
shows that the changes of spin dyad preserving the Hermitian structure
induced by the Hermitian spinor $\varpi_{AA'}$ are of the form
$\grave{\epsilon}_{\bf A}{}^A = O_{\bf A}{}^{\bf B} \epsilon_{\bf
  B}{}^A$ where $(O_{\bf A}{}^{\bf B})$ are $SU(2,\mathbb{C})$
matrices. As $SU(2,\mathbb{C})$ is a subgroup of $SL(2,\mathbb{C})$,
one has that $\epsilon_{\bf AB} = O_{\bf A}{}^{\bf C} O_{\bf B}{}^{\bf
  D}\epsilon_{\bf CD}$. The matrices $( O_{\bf j}{}^{\bf i})$ and
$(O_{\bf A}{}^{\bf B})$ are related to each other via the spatial
Infeld-van der Waerden symbols:
\[
O_{\bf i}{}^{\bf j} = \sigma_{\bf i}{}^{\bf AB} \sigma^{\bf j}{}_{\bf
  CD}  O_{\bf A}{}^{\bf C} O_{\bf B}{}^{\bf D}.
\]

\end{document}